\documentclass[a4paper,11pt]{article}

\usepackage[latin1]{inputenc}
\usepackage[T1]{fontenc}
\usepackage{lmodern}

\usepackage{setspace} 
\usepackage{natbib}
\renewcommand{\theenumi}{\roman{enumi}}
\usepackage{enumitem}\setlist[enumerate]{label={\rm(\theenumi)}}

\usepackage{graphicx} 
\usepackage{subfig}

\usepackage{amsmath,amssymb,amsthm}
\usepackage{mathtools} 
\usepackage{accents} 
\usepackage{subdepth}
\usepackage{mathrsfs} 

\numberwithin{equation}{section}

\theoremstyle{plain}
\newtheorem{theorem}{Theorem}[section]
\newtheorem{proposition}[theorem]{Proposition}
\newtheorem{corollary}[theorem]{Corollary}
\newtheorem{lemma}[theorem]{Lemma}

\theoremstyle{definition}

\theoremstyle{remark}
\newtheorem{remark}[theorem]{Remark}

\newtheoremstyle{key}{3pt}{3pt}{}{}{\itshape}{:}{.5em}{}
\theoremstyle{key}

\DeclareMathOperator*{\esssup}{ess\,sup}
\DeclareMathOperator*{\essinf}{ess\,inf}

\renewcommand{\mid}{\,\vert\,}
\newcommand{\bigv}{\!\bigm\vert\!}
\newcommand{\Bigv}{\!\Bigm\vert\!}
\newcommand{\biggv}{\!\biggm\vert\!}

\newcommand{\nbd}[1]{$#1$\nobreakdash-\hspace{0pt}}
\newcommand{\indi}[1]{\mathbf{1}_{#1}}

\newcommand{\N}{\mathbb{N}}

\newcommand{\R}{\mathbb{R}}

\newcommand{\F}{{\mathscr F}}

\newcommand{\T}{{\mathscr T}}

\newcommand{\cP}{{\cal P}}

\newcommand{\ubar}[1]{\underaccent{\bar}{#1}}

\title{Preemptive Investment under Uncertainty}
\author{Jan-Henrik Steg\thanks{Center for Mathematical Economics, Bielefeld University, Germany. \texttt{jsteg@uni-bielefeld.de}
\protect\\
Financial support by the German Research Foundation (DFG) via grant Ri 1142-4-2 is gratefully acknowledged.
}}
\date{\small This version: May 28, 2016}

\begin{document}
\maketitle

\begin{abstract}
This paper provides a general characterization of subgame perfect equilibria for strategic timing problems, where two firms have the (real) option to make an irreversible investment. Profit streams are uncertain and depend on the market structure. The analysis is based directly on the inherent economic structure of the model. In particular, the determination of equilibria with preemptive investment is reduced to solving a single class of constrained optimal stopping problems. The general results are applied to typical state-space models, completing commonly insufficient equilibrium arguments, showing when uncertainty leads to qualitatively different behavior, and establishing additional equilibria that are Pareto improvements.


\end{abstract}

\section{Introduction}\label{sec:intro}

Preemption is a well-known phenomenon in the context of irreversible investment. In their seminal paper, \cite{FudenbergTirole85} argue that the commitment power of irreversibility and subgame perfectness together imply that any firm which is the first to adopt a new technology in some industry can deter adoption by another firm; the second adopter's benefits will be reduced by competition and thus not worth the immediate adoption cost. In consequence, the firms try to preempt each other to win the (temporary) monopoly profit.\footnote{%
This effect does not appear in simple Nash equilibria as studied by \cite{Reinganum81}, where firms precommit to adoption times.
}

Such preemption is particularly remarkable when it is costly. In their deterministic model, \cite{FudenbergTirole85} assume that the adoption cost decreases over time, which generates an incentive to delay adoption and thus a conflict with the preemption impulse. Another possibility is to introduce uncertainty, so that the option value of investing only in sufficiently good states would make the firms want to wait. There is already a sizable literature on such real-option games that argues for the drastic impact of competition on the valuation of real options, typically using the principles of \cite{FudenbergTirole85}.

However, transferring arguments from the deterministically evolving ``state'' time to a state with stochastic dynamics is only possible to some extent and in fact the arguments often remain incomplete. Analogies are typically drawn for certain reduced-form value functions on the state space, which standard Markovian state dynamics conveniently allow to derive, but current values at different states convey less information than discounted values. Discounting is simple in a deterministic model but difficult to visualize if it is random which critical state is reached next and when.

Figure \ref{fig:FT} shows the discounted (to time $t=0$) values of the firms in the model of \cite{FudenbergTirole85} if the first adoption happens at $t\geq 0$. If a single firm is the first to adopt, its value is $L(t)$ and that of the other firm $F(t)$; the value of simultaneous adopters is $M(t)$. The strategic structure is quite clear: Initially it is optimal to wait, to benefit from the increase in $L$ if the opponent does not adopt and from $F>L>M$ else, then there is a phase with first-mover advantage $L>F$ that may induce preemption, and all payoffs are eventually identical and decreasing, so that adoption becomes dominant. 

\begin{figure}[ht!]%
  \centering
  \subfloat{\includegraphics{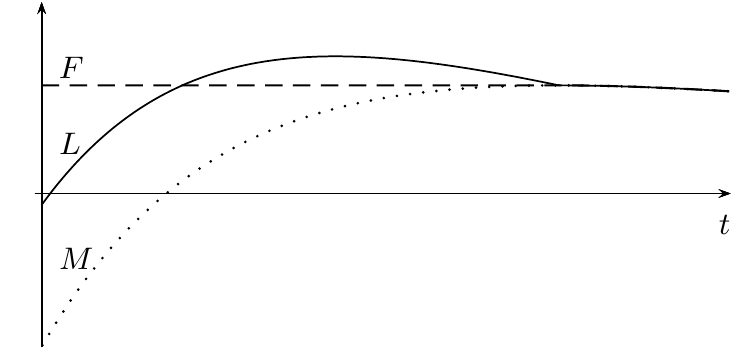}}
  \subfloat{\includegraphics{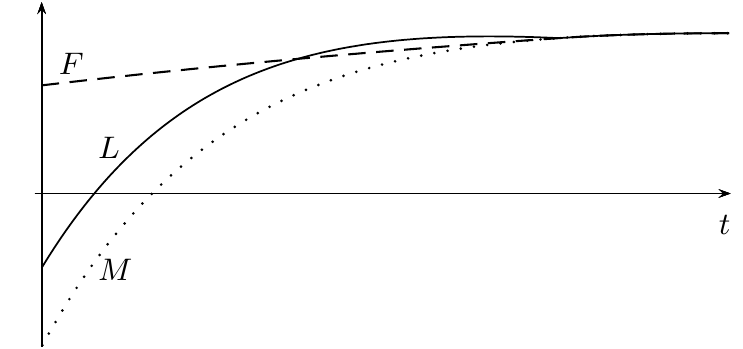}}
  \caption{Two parameterizations of the model of \cite{FudenbergTirole85}.\protect\footnotemark}
  \label{fig:FT}
\end{figure}
\footnotetext{Left: $\pi_0(0)=\pi_0(1)=0$, $\pi_1(1)=0.03$, $\pi_1(2)=0.012$, $r=0.02$, $c(t)=e^{-(r+a)t}$, $a=0.08$. Right: same, except $\pi_0(0)=0.006$, $\pi_1(1)=0.022$.}

\begin{figure}[ht!]%
  \centering
  \subfloat{\includegraphics{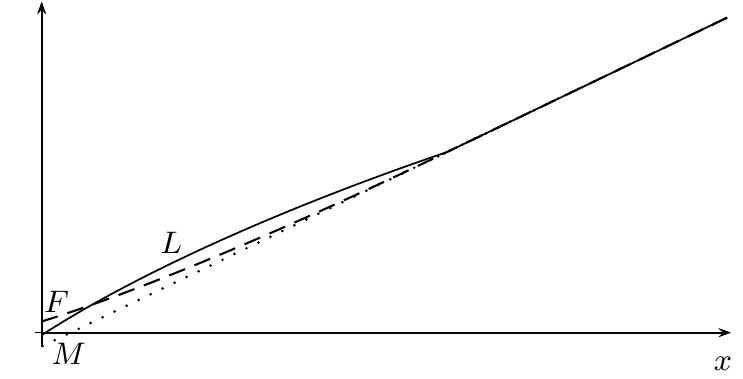}}
  \subfloat{\includegraphics{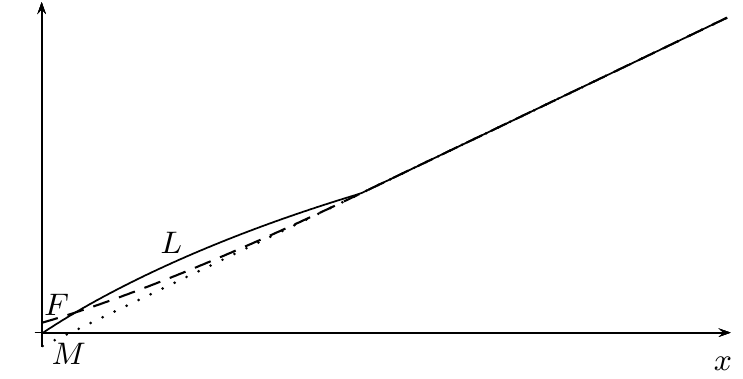}}
  \caption{A typical stochastic model and its deterministic limit.\protect\footnotemark}
  \label{fig:GBM}
\end{figure}
\footnotetext{Model from Section \ref{subsec:PK} with $D_{00}=D_{01}=0$, $D_{10}=2.5$, $D_{11}=I^1=I^2=1$, $r=0.1$, $\mu=0.08$, $\sigma=0.2$ (left) and $\sigma=0$ (right); cf.\ also fn.\ \ref{fn:nest} on the deterministic limit.}

Only the same local orders can be seen in Figure \ref{fig:GBM}, showing \emph{current} values as functions of the stochastic state $x$ for a typical model, but not the dynamics of discounted values. In the right panel the volatility is set to zero, so one can add discounting while $x$ grows according to its law over the shown range, which yields again the left panel of Figure \ref{fig:FT}. After discounting, e.g.\ $F$ is concave, although it is convex in the state. Figure \ref{fig:GBM} reveals the structure of the state space concerning first- or second-mover advantages, but for a complete equilibrium analysis the dynamics of discounted values need to be studied in greater detail than is often the case.

Here we formulate a strategic investment model based on revenue streams that keeps the stochastic dimension completely general, to establish the structure of subgame perfect equilibria by elementary arguments with immediate economic meaning. By directly comparing revenue streams and implied opportunity costs, the verification of equilibria with preemption is reduced to solving a single class of non-strategic optimal stopping problems for one firm. Thereby, not only incomplete arguments for typical equilibria are amended, but the unified view also provides more detailed economic insights into their structure; many economically quite diverse models from the literature can be nested. As mutual preemption destroys value, we also establish some principles in the general framework for when it can be avoided, and we identify times when it is impossible to delay investment in equilibrium.

Alongside, important general questions for equilibria of real-option games are addressed, such as:
\begin{itemize}
\item
At what times is there a first-mover advantage for both firms that they may fight for by trying to preempt each other?
\item
When and how is the first investment affected by a threat of preemption?
\item
Will a firm ever want to invest when it has a second-mover advantage?
\end{itemize}
Answers to these questions will be found by studying appropriate optimal stopping problems.

By applying the general results to two typical state-space models from the literature, those of \cite{Grenadier96} and \cite{PawlinaKort06}, we show how they complete insufficient equilibrium arguments for these and similar other models, and we identify some neglected equilibrium behavior that can qualitatively distinguish such stochastic models from deterministic ones. The common behavior in equilibrium is that the first investment occurs once the state reaches a critical value, such as a certain level of demand. That threshold depends on strategic considerations, but there is no risk to wait for it as no firm can find it profitable to invest before. More complex behavior can result for higher state levels that are too profitable to deter a follower but where profitability is still growing in the state. Then waiting for the most profitable state bares a risk in stochastic models, that of investment due to a first-mover advantage at lower levels, implying feedback effects. We further identify some neglected equilibria for each model that may be Pareto improvements and thus more plausible.

More generally, some models that can be nested here are the deterministic ones of \cite{Reinganum81} and \cite{FudenbergTirole85}, the stochastic model of \cite{MasonWeeds10}, where revenue is linear in a geometric Brownian motion, as in the model of \cite{PawlinaKort06}, who add asymmetry in investment costs, which is further extended to an exponential L\'evy process by \cite{BoyarchenkoLevendorskii14}; the model of \cite{Weeds02} includes Poisson arrivals of R{\&}D success and the model of \cite{Grenadier96} includes a construction delay, but they are both formally equivalent to a symmetric setting with geometric Brownian motion again. 

The paper is organized as follows. The general model is presented in Section \ref{sec:model}. Section \ref{sec:eql} characterizes equilibria with and without preemption, providing sufficient and some necessary equilibrium conditions. The implications for typical state-space models are illustrated in Section \ref{sec:example}. Further specific arguments for such models help to analyze their equilibria in detail, in particular so far neglected aspects. Section \ref{sec:conc} concludes. Some technical results are collected in Appendix \ref{app:add} and all other proofs in Appendix \ref{app:proof}.

\section{The model}\label{sec:model}

Consider two firms $i\in\{1,2\}$ that each can choose when to make one irreversible investment. For instance, firm $i$ may wish to enter some new market, or to improve present operations by updating technology or expanding production capacity. If the firms' markets are related or even the same, then each firm's investment has a potential effect on both firms' revenues. Therefore assume that as long as no firm has invested, the revenues that any firm $i$ incurs per period are given by some stochastic process $(\pi^{0i}_t)$. If firm $i$ invests before its opponent, its revenues switch to the process $(\pi^{Li}_t)$, whereas if the opponent invests first, firm $i$'s revenues switch to the process $(\pi^{Fi}_t)$. Once both firms have invested, firm $i$'s revenues follow the process $(\pi^{Bi}_t)$. The revenues $\pi^{Li}_\cdot$ and $\pi^{Bi}_\cdot$ that apply after firm $i$'s investment are understood net of any capitalized investment cost.\footnote{%
Any investment cost that is strictly decreasing in discounted terms, like $c(t)$ in \cite{FudenbergTirole85}, can be capitalized by a change of variable defined by $c(t)=e^{-ry}=\int_y^\infty e^{-rz}r\,dz$. If the discounted investment cost is stochastic and strictly decreasing in expectation (a strict supermartingale), then one can use the monotone part of the Doob-Meyer decomposition in place of $c(t)$. 
}
All revenues are already discounted to time $0$ units.

Time is continuous, $t\in\R_+$, so only accrued revenues in intervals of time matter. Therefore assume the revenues to be product-measurable w.r.t.\ a given probability space $\bigl(\Omega,\F,P\bigr)$ and the time domain. Assume them in fact to be $P\otimes dt$-integrable, i.e.\ $E\bigl[\int_0^\infty\bigl\lvert\pi^{0i}_t\bigr\rvert\,dt\bigr]<\infty$ etc., to ensure finite expectations throughout. Any (in-)equalities between revenue processes are correspondingly understood to hold \nbd{P\otimes dt}a.e.~-- and any between random variables \nbd{P}a.s. 

Dynamically revealed information about the state of the world is represented by a filtration $\mathbb{F}=(\F_t)$ satisfying the \emph{usual conditions} of right-continuity and completeness. Assume that the past revenues (potentially) accrued up to any $t\in\R_+$ are $\F_t$-measurable, i.e.\ the processes $(\int_0^t\pi^{0i}_s\,ds)$ etc.\ are \emph{adapted} to $\mathbb{F}$.\footnote{\label{fn:pm}%
This property holds e.g.\ if the processes $(\pi^{0i}_t)$ etc.\ are \emph{progressively measurable}, i.e.\ if the restricted mappings $\pi^{0i}_\cdot\colon\Omega\times[0,T]\to\R$ etc.\ are $\F_T\otimes\mathcal{B}([0,T])$-measurable for all $T\in\R_+$.
}

As a further economic assumption, the following orders are imposed on the revenues. First, any investment by a single firm rather harms the revenue of the opponent. For both $i=1,2$ let thus $\pi^{Li}_\cdot\geq\pi^{Bi}_\cdot$ (e.g.\ as the first investor loses a monopoly premium when the laggard invests) and also $\pi^{0i}_\cdot\geq\pi^{Fi}_\cdot$ (e.g.\ as the first investor steals some business from the laggard). The special case $\pi^{0i}_\cdot\equiv\pi^{Fi}_\cdot$ is typical for market entry models and has additional implications that will be pointed out frequently.

Second, firm $2$ rather has a disadvantage in the sense of smaller investment gains relative to being laggard, formally $\pi^{B2}_\cdot-\pi^{F2}_\cdot\leq\pi^{B1}_\cdot-\pi^{F1}_\cdot$ and $\pi^{L2}_\cdot-\pi^{F2}_\cdot\leq\pi^{L1}_\cdot-\pi^{F1}_\cdot$. That disadvantage arises e.g.\ from a higher capitalized investment cost. Given the first part of the disadvantage, that firm $2$'s investment gain as laggard is at most that of firm $1$, the second part would also obtain if $\pi^{Li}_\cdot-\pi^{Bi}_\cdot$ was not greater for firm $2$ than for firm $1$, which is the first investor's revenue loss due to the laggard's investment.

\subsection{The investment timing game}\label{subsec:game}

The firms' investment timing decisions are strategic if some firm's investment indeed affects the other firm's payoff, i.e.\ if $\{\pi^{Li}_\cdot>\pi^{Bi}_\cdot\}$ or $\{\pi^{0i}_\cdot>\pi^{Fi}_\cdot\}$ have positive measure for some $i\in\{1,2\}$. We model the problem as a dynamic game in continuous time. 

As continuous time is not well ordered, it is not possible to define consistent outcomes if one lets the firms choose between the actions ``invest'' and ``wait'' at all times as in a game in extensive form, unless one adds restrictions such as reaction lags (see \citealp{SimonStinchcombe89}, or \citealp{Alos-FerrerRitzberger08}). We follow the typical approach for timing games and let the firms form ``plans of action'', which are dates when to invest \emph{if no other firm invests before}. Any firm whose plan is minimal invests at that planned date, which thus resolves the strategic conflict of who invests first (cf.\ \citealp{FudenbergTirole85}, or \citealp{Larakietal05}). The actual investment date of a firm whose plan is not minimal is determined conditionally at the first investor's date, as an  optimal reaction to the changed history. 

A plan of action is pursued as long as the trivial action history is observed, that no firm has invested. Here plans may be contingent on other dynamic information: about the uncertain state. The fundamental mathematical concept to determine a state-depend date dynamically by the filtration $\mathbb{F}$ is a \emph{stopping time}, a random variable $\tau\colon\Omega\to\R_+\cup\{\infty\}:=[0,\infty]$ satisfying $\{\tau\leq t\}\in\F_t$ for all $t\in\R_+$. The information at any time $t$ must tell if the date $\tau$ (the choice ``stop'') has been reached or not. Let $\T$ denote the set of all stopping times and thus plans. 

Plans of action also have to be formed off the path of play to support subgame perfectness, in particular past an initial plan. Therefore we also consider any stopping time $\vartheta\in\T$ as a potential date at which no firm has invested, yet, and thus as the beginning of a subgame for which plans of action are needed \citep[cf.][]{RiedelSteg14}. The latter cannot be assembled in a measurable way from ($\F_t$-measurable) plans at deterministic times $t\in\R_+$, unless $\F$ is countable. Nevertheless, plans shall be consistent.

A \emph{strategy} for firm $i$ in the timing game is thus a family of plans $\bigl\{\tau_\vartheta^i,\vartheta\in\T\bigr\}$, which is required to satisfy the feasibility condition $\tau_\vartheta^i\geq\vartheta$ for all $\vartheta\in\T$, i.e.\ a plan cannot date in the past, and the time consistency condition $\tau_\vartheta^i=\tau_{\vartheta'}^i$ on the event $\{\vartheta'\leq\tau_\vartheta^i\}$ for any two $\vartheta\leq\vartheta'\in\T$, i.e.\ a plan is not revised while it is not reached, yet. In particular, if the beginnings of two subgames agree in some states, then the plans must agree there, too.

\subsection{Payoffs and equilibrium}\label{subsec:payoff}

To determine an optimal reaction for a firm whose plan is not minimal and thus the resulting conditional payoffs at the time of the first investment, suppose that the opponent of firm $i$ is the first investor at arbitrary $\tau\in\T$. Then firm $i$ may invest at any stopping time $\tau'\geq\tau$, aiming to attain the conditional \emph{follower} payoff
\begin{align}\label{F}
F^i(\tau)&=\int_0^\tau\pi^{0i}_s\,ds+\esssup_{\tau'\geq\tau}E\biggl[\int_\tau^{\tau'}\pi^{Fi}_s\,ds+\int_{\tau'}^\infty\pi^{Bi}_s\,ds\biggv\F_\tau\biggr].\protect\footnotemark 
\end{align}
\footnotetext{%
A random variable is measurable w.r.t.\ $\F_\tau=\{A\in\F\mid\forall t\in\R_+\colon A\cap\{\tau\leq t\}\in\F_t\}$ if its value is known whenever $\tau$ has occured. The value of a stochastic process at $\tau$ is an $\F_\tau$-measurable random variable if the process is progressively measurable (cf.\ fn.\ \ref{fn:pm}), which holds for $(\int_0^t\pi^{0i}_s\,ds)$ by adaptedness and path continuity.
}%
By continuity and integrability of the process $\int_\tau^\cdot\pi^{Fi}_s\,ds+\int_\cdot^\infty\pi^{Bi}_s\,ds$ to be stopped, there exists a \emph{latest} optimal~-- thus uniquely defined~-- stopping time $\tau^i_F(\tau)\in\T$ attaining the value $F^i(\tau)$.

Now suppose on the contrary that firm $i$ is the first investor at $\tau\in\T$. Then the other firm $j\in\{1,2\}\setminus\{i\}$ is assumed to follow suit at $\tau^j_F(\tau)$ to realize $F^j(\tau)$, which yields firm $i$ the conditional \emph{leader} payoff
\begin{equation}\label{L}
L^i(\tau)=\int_0^\tau\pi^{0i}_s\,ds+E\biggl[\int_\tau^{\tau^j_F(\tau)}\pi^{Li}_s\,ds+\int_{\tau^j_F(\tau)}^\infty\pi^{Bi}_s\,ds\biggv\F_\tau\biggr].
\end{equation}
Finally, if both firms invest simultaneously at $\tau\in\T$, firm $i$'s conditional payoff is
\begin{equation}\label{M}
M^i(\tau)=\int_0^\tau\pi^{0i}_s\,ds+E\biggl[\int_\tau^\infty\pi^{Bi}_s\,ds\biggv\F_\tau\biggr]\leq\min\bigl\{F^i(\tau),L^i(\tau)\bigr\}.
\end{equation}
In particular, if no firm invests in finite time, then firm $i$ realizes the payoff
\[
M^i(\infty)=\int_0^\infty\pi^{0i}_s\,ds=F^i(\infty)=L^i(\infty).
\]

\begin{remark}[Regularity of the payoff processes]
By Lemma \ref{lem:LFreg} in Appendix \ref{app:tech} there are processes $(L^i_t)$, $(F^i_t)$ and $(M^i_t)$, that, if evaluated at any stopping time $\tau\in\T$, yield the right-hand side of \eqref{F}, \eqref{L} and \eqref{M}, respectively. We will use them in the following as it is much more convenient to work with processes than families like $\{F^i(\tau),\tau\in\T\}$. Indeed, by Lemma \ref{lem:LFreg} we may assume those processes to have right-continuous paths and thus to be well measurable. The payoffs are also sufficiently integrable to be bounded in expectation and such that pathwise limits at any stopping time induce the corresponding limit in expectation.
\end{remark}

Given two plans $\tau_\vartheta^1,\tau_\vartheta^2\in\T$ for the subgame beginning at $\vartheta\in\T$ with no firm having invested, yet, the first investment happens at $\min\{\tau_\vartheta^1,\tau_\vartheta^2\}$ and firm $i$ becomes leader where $\tau_\vartheta^i<\tau_\vartheta^j$, follower where $\tau_\vartheta^i>\tau_\vartheta^j$, and otherwise simultaneous investment occurs. Thus, given the conditional payoffs at the first investment, firm $i$'s conditional expected payoff at $\vartheta$ is
\begin{align}\label{stop_i}
E\Bigl[L^i_{\tau_\vartheta^i}\indi{\tau_\vartheta^i<\tau_\vartheta^j}+F^i_{\tau_\vartheta^j}\indi{\tau_\vartheta^i>\tau_\vartheta^j}+M^i_{\tau_\vartheta^i}\indi{\tau_\vartheta^i=\tau_\vartheta^j}\Bigv\F_\vartheta\Bigr].
\end{align}
Obviously firm $i$ can only become leader before any given $\tau_\vartheta^j$; otherwise it will at most become follower at $\tau_\vartheta^j$. 

Two strategies $\bigl\{\tau_\vartheta^1,\vartheta\in\T\bigr\},\bigl\{\tau_\vartheta^2,\vartheta\in\T\bigr\}$ are a \emph{subgame perfect equilibrium} if there is no $i\in\{1,2\}$ and $\vartheta\in\T$ such that \eqref{stop_i} can be increased with positive probability by choosing any stopping time $\tau\geq\vartheta$ instead of $\tau_\vartheta^i$, i.e.\ if any two plans $\tau_\vartheta^1,\tau_\vartheta^2$ are an \emph{equilibrium} for the subgame beginning at $\vartheta$.

\section{Equilibrium characterization}\label{sec:eql}

The assumed orders between different revenues have important consequences for equilibria of the timing game, independently of any more specific model for the uncertainty. This section illuminates the structure of possible equilibria just by comparing revenue streams, to provide more detailed economic insights than analyses based on reduced functional forms of the payoffs for specific state-space models, and to provide complete equilibrium arguments. We show that it suffices to solve one particular class of constrained optimal stopping problems to construct subgame perfect equilibria with preemption about any first-mover advantage. As mutual preemption may destroy option values unnecessarily, we then identify times past which investment cannot be delayed in any equilibrium. We finally consider alternative equilibria that avoid preemption and provide arguments simplifying their verification. 


\subsection{Sufficient equilibrium conditions}\label{subsec:sufeql}

In order to construct subgame perfect equilibria, it is first determined in which subgames immediate investment is an equilibrium, possibly due to a mutual preemption scheme.

\subsubsection{Simultaneous investment}\label{subsec:follower}

Immediate investment by both firms is an equilibrium at $\vartheta\in\T$ if both firms' follower options are worthless, i.e.\ if $F^i_{\vartheta}=M^i_{\vartheta}$ for both $i=1,2$. If a firm $i$ deviated to any plan $\tau_\vartheta^i>\vartheta$, it would become follower and actually invest at $\tau_F^i(\vartheta)$, which still attains $F^i_{\vartheta}$. In particular, if $\vartheta=\tau^i_F(\vartheta)$ for both $i=1,2$, then a unilateral deviation from simultaneous investment would not even change the physical outcome and both firms $i$ still obtain $L^i_{\vartheta}=F^i_{\vartheta}=M^i_{\vartheta}$. Note that even in this case~-- when any follower would merely forego profitable revenue by hesitating~-- either firm $i$ may nevertheless only be willing to invest \emph{proactively} by the plan $\tau_\vartheta^i=\vartheta$ because the other firm does so. If a firm $i$'s investment was only \emph{triggered} by $\tau^i_F(\vartheta)=\vartheta<\tau_\vartheta^i$, then the opponent might want to delay investment (see Section \ref{subsec:neceql} below).

Given the assumption $\pi^{B2}_\cdot-\pi^{F2}_\cdot\leq\pi^{B1}_\cdot-\pi^{F1}_\cdot$, firm $1$'s follower option is not more valuable than firm $2$'s, so simultaneous investment is an equilibrium at $\vartheta\in\T$ if $\tau'=\vartheta$ attains $F^2_\vartheta$. Similarly, firm $1$'s follower reaction time will never exceed firm $2$'s.

\begin{lemma}\label{lem:tau^1_F<tau^2_F}
$\tau^1_F(\tau)\leq\tau^2_F(\tau)$ and $F^1_{\tau}-M^1_{\tau}\leq F^2_{\tau}-M^2_{\tau}$ for any $\tau\in\T$.
\end{lemma}
%
%
Lemma \ref{lem:tau^1_F<tau^2_F} is based on the fact that a follower's opportunity cost of waiting is $\pi^{Bi}_\cdot-\pi^{Fi}_\cdot$, which is not less for firm $1$ than for firm $2$. Thus, if firm $1$ is follower, it cannot wait longer than firm $2$ could. More generally, firm $1$ cannot gain more from waiting until any time than firm $2$ could, so firm $1$'s option value $F^1_\tau-M^1_\tau$ as follower is at most what firm $2$'s would be.

\subsubsection{Preemption}\label{subsec:preem}

Critical phases of a timing game are when both players have a first-mover advantage, i.e.\ the set $\cP:=\{L^1_\cdot>F^1_\cdot\}\cap\{L^2_\cdot>F^2_\cdot\}\subset\Omega\times\R_+$. If any player plans to become leader in such a phase, e.g.\ because there is no subsequent continuation equilibrium promising at least the same payoff in expectation, then a preemption scheme is triggered with both players trying to stop waiting before each other to become leader. Therefore $\cP$ will be called \emph{preemption region}. If simultaneous stopping is not an equilibrium, and if each player would prefer to wait without preemptive pressure, then there may be no equilibrium at all, not even for deterministic, very regular models and considering mixed strategies \citep[see e.g.][]{FudenbergTirole85}.

As the players cannot plan to stop ``immediately after'' each other in continuous time, additional outcomes have to be facilitated for plans to stop at the same date. The aim is to sustain the preemption effect that if anyone hesitates, the respective other player becomes leader. A player prefers to ``exert pressure'' by planning to stop at the same date as the other rather than later if the former yields at least the follower payoff in expectation.

For modeling simplicity we here assume that if both firms plan to invest at any first hitting time of the preemption region $\tau_\cP(\vartheta):=\inf\{t\geq\vartheta\mid (L^1_t>F^1_t)\wedge(L^2_t>F^2_t)\}\in\T$, then this is due to mutual preemption, and thus each firm obtains its follower payoff. These expected payoffs arise from a suitable distribution over who invests first at $\tau_\cP(\vartheta)$, i.e.\ firm $1$, firm $2$ or both. Then any firm $i$ becomes leader (the best outcome, by right-continuity of $L^i_\cdot-F^i_\cdot$), follower or simultaneous investor (the worst outcome) with respective probabilities. Such a distribution can be endogenized by extending the strategy spaces to capture outcome limits from mixed strategies in discrete time; see \cite{FudenbergTirole85} for deterministic models and \citet[Proposition 3.1]{RiedelSteg14} for a generalization to stochastic models. An exception is made if one firm is indifferent to become leader or follower at $\tau_\cP(\vartheta)$; then the other firm becomes leader. Now the plans $\tau_\vartheta^1=\tau_\vartheta^2=\vartheta$ are an equilibrium where $\vartheta=\tau_\cP(\vartheta)$.

Here, given $\tau^1_F(\cdot)\leq\tau^2_F(\cdot)$ and the assumption $\pi^{L1}_\cdot-\pi^{F1}_\cdot\geq\pi^{L2}_\cdot-\pi^{F2}_\cdot$, firm $1$'s first-mover advantage is never less than firm $2$'s, so $\cP=\{L^2_\cdot>F^2_\cdot\}$ and $\tau_\cP(\vartheta)=\inf\{t\geq\vartheta\mid L^2_t>F^2_t\}$. 

\begin{lemma}\label{lem:L-F}
$L^1_\tau-F^1_\tau\geq L^2_\tau-F^2_\tau$ for any $\tau\in\T$.
\end{lemma}
%
%
Lemma \ref{lem:L-F} uses the fact that the revenue difference between being leader or follower is $\pi^{Li}_\cdot-\pi^{Fi}_\cdot$ until any follower would invest, which is not less for firm $1$ than for firm $2$. Firm $1$ further prefers to be leader between its own follower reaction time and that of firm $2$, because it earns $\pi^{L1}_\cdot$ instead of $\pi^{B1}_\cdot$. Firm $2$, on the contrary, cannot gain from being leader between those two times, as it can only realize $\pi^{B2}_\cdot$ instead of $\pi^{F2}_\cdot$, which a follower never prefers before its own reaction time.

Firm $2$ can only have a first-mover advantage when $\pi^{B1}_\cdot-\pi^{F1}_\cdot$ is not too profitable for firm $1$: if $\vartheta=\tau^1_F(\vartheta)$, then $L^2_\vartheta=M^2_\vartheta\leq F^2_\vartheta$.\footnote{%
If $\vartheta=\tau^1_F(\vartheta)$, then it is indeed not even on the boundary of $\cP$ if $\tau'=\vartheta$ does not attain $F^2_\vartheta$, as then $L^2_\vartheta=M^2_\vartheta<F^2_\vartheta$ and hence $\vartheta<\tau_\cP(\vartheta)$ by right-continuity of the processes. 
}
However, investment must be sufficiently profitable in terms of the revenue difference $\pi^{L2}_\cdot-\pi^{F2}_\cdot$~-- firm $2$'s potential gain from being leader instead of follower~-- and firm $2$ can in fact only have a first-mover advantage if it still does at the optimal times to start $\pi^{L2}_\cdot-\pi^{F2}_\cdot$, see Appendix \ref{app:locpreem}. As $\pi^{B2}_\cdot-\pi^{F2}_\cdot\leq\pi^{B1}_\cdot-\pi^{F1}_\cdot$, $\pi^{L2}_\cdot$ must exceed $\pi^{B2}_\cdot$ enough. In particular, $\cP=\emptyset$ if $\pi^{L2}_\cdot-\pi^{F2}_\cdot\leq\pi^{B1}_\cdot-\pi^{F1}_\cdot$, because then firm $1$ would follow immediately at the latest optimal time to start $\pi^{L2}_\cdot-\pi^{F2}_\cdot$. The latter nonstrategic stopping problem is firm $2$'s monopoly problem \eqref{Mono^i} considered below if $\pi^{02}_\cdot\equiv\pi^{F2}_\cdot$, like in typical market entry models. For state-space models that admit threshold-type solutions for these problems, it may even suffice to look at one threshold to see if $\cP=\emptyset$ (instead of at a whole half-space where stopping is optimal), like for the applications in Section \ref{sec:example}.

\subsubsection{Subgame perfect equilibrium with preemption}\label{subsec:eqlstop}

The subsequent equilibrium construction is facilitated by the fact that independently of what happens in the preemption region, no firm ever wants to invest when it has a second-mover advantage. This finding is driven by the assumption that investment does not benefit the other firm. In contrast to some suggestions in the literature, a second-mover advantage alone does not suffice to delay investment in general.

\begin{lemma}\label{lem:L<Fwait}
Investment is never optimal for any firm $i\in\{1,2\}$ where $F^i_\cdot>L^i_\cdot$. Further, waiting until $\min\bigl\{\tau_\cP(\vartheta),\tau^2_F(\vartheta)\bigr\}$ does not restrict firm $2$'s payoff in the subgame at $\vartheta\in\T$ for any (even mixed) strategy of firm $1$. 
\end{lemma}
%
%
Where $F^i_\cdot\geq L^i_\cdot$, firm $i$ can realize at least the preferred follower payoff in expectation by planning to invest at its follower reaction time. Indeed, the follower payoff is nondecreasing in expectation (a \emph{submartingale}) until that time~-- if the opponent invests in the meantime, that does not affect firm $i$'s reaction and can only defer the laggard revenue $\pi^{Fi}_\cdot\leq\pi^{0i}_\cdot$~-- and at the own reaction time, investing regardlessly is at least as good as becoming follower.

By Lemma \ref{lem:L<Fwait} we may let firm $2$'s plan for $\vartheta$ be to invest at $\min\bigl\{\tau_\cP(\vartheta),\tau^2_F(\vartheta)\bigr\}$, where preemption or simultaneous investment is an equilibrium. In case of symmetric revenues, the same plan is then a best reply for firm $1$, but in general, firm $1$ may have a strict first-mover advantage before $\tau_\cP(\vartheta)$ and may want to exploit it. Given the preemption payoffs from Section \ref{subsec:preem} at $\tau_\cP(\vartheta)$ and $L^1_\cdot=F^1_\cdot=M^1_\cdot$ at $\tau^2_F(\vartheta)$, firm $1$ can realize $L^1_\cdot$ anywhere before or at $\min\bigl\{\tau_\cP(\vartheta),\tau^2_F(\vartheta)\bigr\}$, except where $L^2_\cdot>F^2_\cdot$ at $\tau_\cP(\vartheta)$: there firm $1$ will get $F^1_\cdot$. As $L^2_\cdot>F^2_\cdot$ in fact only at $\tau_\cP(\vartheta)$, the best reply problem for firm $1$ at any $\vartheta\in\T$ is
\begin{align}
&\esssup_{\vartheta\leq\tau\leq\tau_\cP(\vartheta)\wedge\tau^2_F(\vartheta)}E\Bigl[L^1_\tau\indi{\{L^2_\tau\leq F^2_\tau\}}+F^1_\tau\indi{\{L^2_\tau>F^2_\tau\}}\Bigv\F_\vartheta\Bigr]. \label{tildeLl}
\end{align}
If problem \eqref{tildeLl} has a solution $\tau^1_*(\vartheta)$, then its value is firm $1$'s equilibrium payoff at $\vartheta$, and that of firm $2$ is $E\bigl[F^2_{\tau^1_*(\vartheta)}\bigv\F_\vartheta\bigr]$, who gets the follower payoff also where $\tau^1_*(\vartheta)=\min\bigl\{\tau_\cP(\vartheta),\tau^2_F(\vartheta)\bigr\}$.

We can summarize as follows.

\begin{theorem}\label{thm:SPE}
If there is a family of solutions $\bigl\{\tau^1_*(\vartheta),\vartheta\in\T\bigr\}$ to \eqref{tildeLl} that satisfies the time-consistency condition for a strategy, then it is firm $1$'s strategy in a subgame perfect equilibrium in which firm $2$ uses the strategy $\bigl\{\tau^2_*(\vartheta),\vartheta\in\T\bigr\}$ given by $\tau^2_*(\vartheta)=\min\bigl\{\tau_\cP(\vartheta),\tau^2_F(\vartheta)\bigr\}$.

If all revenues are symmetric, then there is a symmetric subgame perfect equilibrium in which both firms use the given strategy of firm $2$.
\end{theorem}

Time consistency can be ensured whenever there are solutions to \eqref{tildeLl}, because then there are respectively earliest ones due to right-continuity.\footnote{%
The families $\bigl\{\tau_\cP(\vartheta),\vartheta\in\T\bigr\}$ and $\bigl\{\tau^2_F(\vartheta),\vartheta\in\T\bigr\}$ satisfy the time-consistency condition by construction and thus also $\bigl\{\tau^2_*(\vartheta),\vartheta\in\T\bigr\}$. As the latter are the constraints in \eqref{tildeLl}, any family of earliest solutions $\bigl\{\tau^1_*(\vartheta),\vartheta\in\T\bigr\}$ will then be time consistent, too.
}
It holds similarly if respectively latest solutions can be chosen and for state-space models if the $\tau^1_*(\vartheta)$ are of threshold-type.

The existence of a solution to \eqref{tildeLl} is generally not clear, however, because the process to be stopped has a discontinuity at $\tau_\cP(\vartheta)$ if $\vartheta<\tau_\cP(\vartheta)<\tau^2_F(\vartheta)$ and $L^2_{\tau_\cP(\vartheta)}>F^2_{\tau_\cP(\vartheta)}$; then also $L^1_{\tau_\cP(\vartheta)}>F^1_{\tau_\cP(\vartheta)}$ by Lemma \ref{lem:L-F} and preemption causes a drop. A solution will exist if the process $L^2_\cdot-F^2_\cdot$ is lower semi-continuous, because then $L^2_{\tau_\cP(\vartheta)}=F^2_{\tau_\cP(\vartheta)}$ on $\{\vartheta<\tau_\cP(\vartheta)\}$, such that \eqref{tildeLl} reduces to 
\begin{align}
&\esssup_{\vartheta\leq\tau\leq\tau_\cP(\vartheta)\wedge\tau^2_F(\vartheta)}E\Bigl[L^1_\tau\Bigv\F_\vartheta\Bigr].\label{Lopt}
\end{align}

\begin{proposition}\label{prop:SPE}
Assume that $L^2_\cdot-F^2_\cdot$ is lower semi-continuous from the left. Then there exists a subgame perfect equilibrium as described in Theorem \ref{thm:SPE}, with each $\tau^1_*(\vartheta)$ the respectively earliest solution of
\begin{equation}\label{MonoPreem}
\esssup_{\vartheta\leq\tau\leq\tau_\cP(\vartheta)\wedge\tau^2_F(\vartheta)}E\biggl[\int_0^{\tau}\pi^{01}_s\,ds+\int_{\tau}^{\infty}\pi^{L1}_s\,ds\biggv\F_\vartheta\biggr].
\end{equation}
\end{proposition}

The solutions of problem \eqref{Lopt} are the~-- by continuity existing~-- solutions of the conceptually much simpler constrained monopoly problem \eqref{MonoPreem} because the follower reaction time $\tau_F^2(\tau)$ in $L^1_\tau$ remains constant for $\tau\in[\vartheta,\tau_F^2(\vartheta)]$, cf.\ also Lemma \ref{lem:stoptau^1_F} below.

In this equilibrium the firms either plan to invest because the respective other does (as soon as both have a first-mover advantage or both would invest as follower), or firm $1$ exploits that waiting is dominant for firm $2$ and it thus acts like a constrained monopolist. 

Only the constraint $\tau\leq\tau_\cP(\vartheta)$ matters in \eqref{MonoPreem} if $\pi^{L1}_\cdot-\pi^{01}_\cdot\geq\pi^{B1}_\cdot-\pi^{F1}_\cdot$, like for market entry with $\pi^{01}_\cdot\equiv\pi^{F1}_\cdot$, because then the solution is to invest no later than at $\tau^1_F(\vartheta)\leq\tau^2_F(\vartheta)$ (cf.\ the discussion after Lemma \ref{lem:stoptau^1_F}). More generally, it is of course optimal to invest in \eqref{MonoPreem} whenever it is so for $i=1$ in the unconstrained monopoly problem
\begin{equation}\label{Mono^i}
\esssup_{\tau\geq\vartheta}E\biggl[\int_0^{\tau}\pi^{0i}_s\,ds+\int_{\tau}^{\infty}\pi^{Li}_s\,ds\biggv\F_\vartheta\biggr].
\end{equation}



\subsection{Necessary equilibrium conditions}\label{subsec:neceql}

In the equilibria derived above, it may often be that investment is only optimal because the other firm plans to invest at the same date. Possibly other equilibria exist with both firms investing later, which then both prefer, but on which they have to coordinate. Now some times are derived when investment is indeed unavoidable in equilibrium.

Equilibria are obviously related to optimal stopping of the leader payoff processes, typically subject to certain constraints, cf.\ \eqref{stop_i}. The next lemma shows that given the assumptions $\pi^{Li}_\cdot\geq\pi^{Bi}_\cdot$ and $\pi^{0i}_\cdot\geq\pi^{Fi}_\cdot$, equilibrium investment must not happen later than when firm $i$ would invest if it had the exclusive right to invest first, i.e.\ if it considered the \emph{unconstrained} problem of when to become leader. 

Due to the dynamic follower reaction in $L^i_\tau$, this is a complex problem. It may for instance not be optimal to invest when the general circumstances are so favorable that any monopolist or follower would invest immediately: When only $\pi^{Bi}_\cdot$ can be realized, it may be better to invest when the follower will react with a lag.\footnote{%
See Remark \ref{rem:optleader} in Appendix \ref{app:tech} on the monopolists' and leaders' problems for standard diffusion models.
}
To become leader optimally, it is however necessary that a monopolist would invest, too.

\begin{lemma}\label{lem:Lmaxstop}
Wherever $\tau=\vartheta$ is the latest stopping time attaining
\begin{equation}\label{maxL}
\esssup_{\tau\geq\vartheta}E\Bigl[L^i_\tau\Bigv\F_\vartheta\Bigr]
\end{equation}
for some $i\in\{1,2\}$, some firm must invest immediately in any equilibrium for the subgame beginning at $\vartheta\in\T$. Further, where $\tau=\vartheta$ attains \eqref{maxL}, it also attains \eqref{Mono^i}.
\end{lemma}
%
%
Lemma \ref{lem:Lmaxstop} rests on the observation that if it is optimal to become leader immediately in \eqref{maxL}, then there is no superior future follower payoff, either: If firm $i$ had the choice when to become follower, it would generally prefer times $\tau^i_F(\tau)$ to avoid the low revenue $\pi^{Fi}_\cdot\leq\pi^{0i}_\cdot$. At any $\tau^i_F(\tau)$, however, becoming follower is not better than becoming leader due to $\pi^{Bi}_\cdot\leq\pi^{Li}_\cdot$. 

The problem \eqref{maxL} becomes much easier by fixing continuation equilibria like simultaneous investment at $\tau^2_F(\vartheta)$ that prevent becoming leader later. By such a constraint, firm $2$'s follower reaction will always be the same and firm $1$ will not cannibalize any revenue $\pi^{L1}_\cdot$ past $\tau^2_F(\vartheta)$ if it invests before. Thus, firm $1$'s leader problem becomes equivalent to a constrained monopolist's problem. For the following constrained version of Lemma \ref{lem:Lmaxstop}, it is also important that firm $1$ will not regret to receive $\pi^{B1}_\cdot$ from $\tau^2_F(\vartheta)$ on by investing before.\footnote{%
Firm $2$, on the contrary, may prefer to become follower at $\tau^1_F(\vartheta)$ and effectively invest later. If firm $2$ can become leader up to $\tau^2_F(\vartheta)$, it may expect a delayed follower reaction and high revenue $\pi^{L2}_\cdot$ in $(\tau^1_F(\vartheta),\tau^2_F(\vartheta)]$ and the problem cannot be simplified.
}

\begin{lemma}\label{lem:stoptau^1_F}
Suppose that firm $2$'s strategy in an equilibrium for the subgame at $\vartheta\in\T$ induces investment no later than at $\tau^2_F(\vartheta)$. Then investment must happen immediately where $\tau=\vartheta$ is the latest stopping time attaining
\begin{equation}\label{maxL_S}
\esssup_{\tau\in[\vartheta,\tau_F^2(\vartheta)]}E\Bigl[L^1_\tau\Bigv\F_\vartheta\Bigr],
\end{equation}
which has the same solutions as
\begin{equation}\label{mono1}
\esssup_{\tau\in[\vartheta,\tau^2_F(\vartheta)]}E\biggl[\int_0^\tau\pi^{01}_s\,ds+\int_{\tau}^{\infty}\pi^{L1}_s\,ds\biggv\F_{\vartheta}\biggr].
\end{equation}
\end{lemma}
%
%
If a monopolist's investment gain $\pi^{L1}_\cdot-\pi^{01}_\cdot$ is not less than a follower's, $\pi^{B1}_\cdot-\pi^{F1}_\cdot$ (like in typical market entry with $\pi^{01}_\cdot\equiv\pi^{F1}_\cdot$), then the latest solution of \eqref{mono1} does not exceed $\tau^1_F(\vartheta)$, where any delay only means foregone revenue for a follower in \eqref{F}, and firm $1$ would now lose no less as prospective leader. In this case \eqref{mono1} has the same solutions as \eqref{Mono^i}.


Another continuation equilibrium that potentially induces earlier investment is preemption at $\tau_\cP(\vartheta)$ as in Section \ref{subsec:preem}. Given preemption in $\cP$ (or $\cP=\emptyset$), firm $2$ can never realize payoffs exceeding $F^2_\cdot$, and the game has to end immediately at all respectively latest optimal times to \emph{become} follower. Indeed, such times have to satisfy $\tau=\tau^2_F(\tau)$ (as it is otherwise no loss to become follower at $\tau^2_F(\tau)$ and receive $\pi^{02}_\cdot\geq\pi^{F2}_\cdot$ longer), and then firm $2$ can enforce the payoff $F^2_\tau=L^2_\tau=M^2_\tau$ by investing regardlessly. 

A stopping time satisfying $\vartheta=\tau^i_F(\vartheta)$ can only maximize firm $i$'s follower payoff if it also maximizes the simultaneous investment payoff. Conversely, an optimal time for simultaneous investment must also be optimal to become follower, as the opportunity cost of waiting for the former, $\pi^{Bi}_\cdot-\pi^{0i}_\cdot$, is at most that for the latter by $\pi^{0i}_\cdot\geq\pi^{Fi}_\cdot$.

\begin{lemma}\label{lem:Fmaxstop}
Every stopping time $\tau_M^i\geq\vartheta$ that attains
\begin{equation}\label{sim_i}
\esssup_{\tau\geq\vartheta}E\Bigl[M^i_\tau\Bigv\F_\vartheta\Bigr]=\esssup_{\tau\geq\vartheta}E\biggl[\int_0^\tau\pi^{0i}_s\,ds+\int_{\tau}^{\infty}\pi^{Bi}_s\,ds\biggv\F_\vartheta\biggr]
\end{equation}
for some given $\vartheta\in\T$ and $i\in\{1,2\}$ also attains
\begin{equation}\label{maxF}
\esssup_{\tau\geq\vartheta}E\Bigl[F^i_\tau\Bigv\F_\vartheta\Bigr].
\end{equation}
If $\tau_M^i\geq\vartheta$ attains \eqref{maxF}, then $\tau_F^i(\tau_M^i)$ also attains \eqref{sim_i}. In particular, the respectively latest solutions of \eqref{sim_i} and \eqref{maxF} agree.
\end{lemma}
%
%
Thus, \eqref{sim_i} and \eqref{maxF} have a latest solution $\tau_M^i\geq\tau^i_F(\vartheta)$. That inequality may be strict in general. If $\pi^{0i}_\cdot\equiv\pi^{Fi}_\cdot$, however, like in typical market entry models, then \eqref{sim_i} equals $F^i_\vartheta$ and $\tau_F^i(\vartheta)$ is the latest time attaining \eqref{maxF}. 


\subsection{Equilibria without preemption}\label{subsec:othereql}

There can be other equilibria without preemption, even if both firms have a strict first-mover advantage at some times, i.e.\ if the region $\cP$ of \emph{potential} preemption is non-empty. Preemption can be avoided by sufficiently profitable continuation equilibria, and this will then constitute a Pareto improvement. For instance, joint investment at a future stopping time $\tau_J$ can be an equilibrium if that yields at least the same expected payoff as becoming leader at any earlier time, like in the right panel of Figure \ref{fig:FT}. The firms can also plan to invest sequentially if one accepts to become follower when the other invests. Such equilibria depend on the relative magnitudes of the revenue processes, however, so existence cannot be ensured by simple regularity properties like continuity in Proposition \ref{prop:SPE}. On the contrary, if $\pi^{Fi}_\cdot\equiv\pi^{0i}_\cdot$, then $F^i_\cdot$ is nonincreasing in expectation (a supermartingale), as becoming follower later only leaves less possibilities to invest optimally. Thus, if $L^i_\vartheta>F^i_\vartheta$, then firm $i$ strictly prefers immediate investment to waiting until firm $j$ invests at some $\tau_j>\vartheta$, because waiting would yield at most $E\bigl[F^i_{\tau_j}\bigv\F_{\vartheta}\bigr]\leq F^i_\vartheta$. 

It is therefore necessary that $\pi^{Fi}_\cdot<\pi^{0i}_\cdot$ occurs (e.g.\ due to the first investment stealing business) to have firm $i$ wait until $\tau_j$ despite earlier first-mover advantages. Then it suffices to check for deviations at very specific times identified in Proposition \ref{prop:devi} in Appendix \ref{app:othereql}, which also avoids to maximize the complex leader payoff directly, but uses simpler problems like \eqref{mono1}. For state-space models, it may even suffice to consider deviations at one threshold, like for the applications in Section \ref{sec:example}.

Proposition \ref{prop:devi} can be applied to candidate times $\tau_J\geq\vartheta$ for joint investment, which for both $i=1,2$ need to satisfy $F^i_{\tau_J}=M^i_{\tau_J}$ and to maximize the expected joint investment payoff $E\bigl[M^i_{\tau_J}\bigv\F_\vartheta\bigr]$ as considered in Lemma \ref{lem:Fmaxstop}, at least up to some constraint. If delayed joint investment is not feasible, then preemption may still be avoidable in an equilibrium with sequential investment. In the equilibria of Theorem \ref{thm:SPE} for $\cP=\emptyset$, firm $1$ becomes leader at an optimal time before simultaneous investment would happen at $\tau_F^2(\vartheta)$. Simply ignoring preemption in a non-empty $\cP$, firm $1$'s problem becomes \eqref{maxL_S}. Any solution $\tau_S\in\T$ of \eqref{maxL_S} or \eqref{mono1} is a best reply for firm $1$ against $\tau_F^2(\vartheta)$. Optimality of the latter for firm $2$ against $\tau_S\leq\tau_F^2(\vartheta)$ can be verified by a further simplification, Corollary \ref{cor:seqeql} in Appendix \ref{app:othereql}. Under an additional revenue order, it suffices to check that $[\tau_S]$ is not in $\cP$.

\section{Applications}\label{sec:example}

As an illustration, the previous general results will now be applied to two typical models from the strategic real options literature, in order to provide complete proofs for basic equilibrium outcomes that are discussed extensively in the literature, to derive neglected equilibria that may constitute Pareto improvements or actually display behavior that qualitatively differs from deterministic models, and to argue that some equilibria analyzed in the literature only exist under additional restrictions, if at all. The model of \cite{PawlinaKort06} first serves as the main vehicle. Here we allow for weak orders among its parameters so that the models of \cite{Weeds02} and \cite{FudenbergTirole85} can be nested.\footnote{\label{fn:nest}%
In \cite{Weeds02}, investment starts an R{\&}D project with success arrival rate $h>0$. The expected payoffs are equivalent to those from \eqref{PKmodel} with augmented discount rate $r+h$ instead of $r$, $D_{00}=D_{01}=0$, $D_{10}=h$, $D_{11}=h(r+h-\mu)/(r+2h-\mu)$ and $I^1=I^2=K$. The model of \cite{FudenbergTirole85} with their concrete discounted cost function $c(t)=e^{-(r+a)t}$ is equivalent to \eqref{PKmodel} with $D_{00}=\pi_0(0)$, $D_{01}=\pi_0(1)$, $D_{10}=\pi_1(1)$, $D_{11}=\pi_1(2)$, $\mu=a$, augmented discount rate $r+a$ instead of $r$ and $\sigma=0$. The solutions in Section \ref{subsec:PK} converge to the solutions for the deterministic case as $\sigma\to 0$. In particular, $\beta_1$ from fn.\ \ref{fn:threshold} increases to $r/(\mu^+)$, so the investment thresholds converge to those for the deterministic case by $\beta_1/(\beta_1-1)\to r/(r-\mu^+)$, as does the expected discount factor for the first time that the state $x_t$ exceeds a threshold $x>x_0$, $(x_0/x)^{\beta_1}$.
}
Afterwards the results of \cite{Grenadier96} will be revisited using the same arguments, although his economic setting is quite different.

\subsection{Irreversible investment with asymmetric costs}\label{subsec:PK}

The model of \cite{PawlinaKort06} is quite prototypic for the real options literature, but the equilibrium analysis is not complete.\footnote{%
Their proposed preemption equilibrium investment, with the high cost firm $2$ investing at its follower threshold $x_F^2$, can only be seen as an outcome, but not as an equilibrium strategy, because firm $1$ is only willing to invest at the preemption point if there is a preemption threat. Equilibrium verification is also incomplete because~-- as in other papers~-- the argument of a current second-mover advantage is insufficient to justify waiting, and only subgames with low initial states are considered despite the aim for subgame perfectness.
}
Theorem \ref{thm:SPE} yields proper subgame perfect equilibria. We will analyze them in detail, to show some remarkable neglected behavior and to make the arguments applicable to other models. The revenue streams for firm $i\in\{1,2\}$ in \cite{PawlinaKort06} are
\begin{equation}\label{PKmodel}
\left.\begin{aligned}
\pi^{0i}_t&=e^{-rt}x_tD_{00},\qquad & \pi^{Li}_t&=e^{-rt}(x_tD_{10}-rI^i), \\ 
\pi^{Fi}_t&=e^{-rt}x_tD_{01}, & \pi^{Bi}_t&=e^{-rt}(x_tD_{11}-rI^i),
\end{aligned}\quad\right\}
\end{equation}
with discount factor $r>0$ and demand uncertainty reflected by a geometric Brownian motion $(x_t)$ satisfying
\begin{equation}\label{GBM}
dx_t=\mu x_t\,dt+\sigma x_t\,dB_t,
\end{equation}
where $(B_t)$ is Brownian noise, $\mu<r$ the expected growth rate and $\sigma>0$ the volatility. The constants $D_{10}\geq D_{11}$ and $D_{00}\geq D_{01}$ capture a negative impact of investment on the opponent's revenue, and $I^2\geq I^1>0$ are the firms' constant investment costs, here capitalized. 
The leader and follower processes are then continuous (as functions of the state $x_t$), and the present instances of the follower problems \eqref{F} and the monopoly problems \eqref{Mono^i} are solved by investing when $x_t$ exceeds some thresholds $x_F^i$ and $x_L^i$, respectively.\footnote{\label{fn:threshold}%
If $D_{11}>D_{01}$, then $x_F^i=\frac{\beta_1}{\beta_1-1}\cdot\frac{I^i(r-\mu)}{D_{11}-D_{01}}$, where $\beta_1>1$ is the positive root of $\frac12\sigma^2\beta(\beta-1)+\mu\beta-r=0$. If $D_{11}\leq D_{01}$, then $x_F^i=\infty$. Analogously, $x_L^i=\frac{\beta_1}{\beta_1-1}\cdot\frac{I^i(r-\mu)}{(D_{10}-D_{00})^{_+}}$. These are standard from option pricing.
}
Thus, simultaneous investment is an equilibrium for all states $x_\vartheta\geq x_F^2$.

If the preemption region in this model is non-empty, it is characterized by an open interval $(\ubar x,\bar x)$ of the state space $\R_+$ with $\bar x\leq x_F^1\leq x_F^2$ (where both inequalities are strict if $I^2>I^1$ and $D_{10}>D_{11}>D_{01}$), such that one can simply call $(\ubar x,\bar x)$ preemption region. The proof of the following proposition generalizes to other models driven by a continuous Markov process that affects revenues monotonically.

\begin{proposition}\label{prop:PKpreemp}
Consider the specification \eqref{PKmodel}. There are two numbers $\ubar x\leq\bar x\in(0,x_F^1]$ such that $L^2_t>F^2_t\Leftrightarrow x_t\in(\ubar x,\bar x)$ for all $t\in\R_+$, with $\bar x=x_F^2$ if $I^1=I^2$.
\end{proposition}
%
%
By Lemma \ref{lem:PreemMono} in Appendix \ref{app:locpreem} and the discussion thereafter it is enough to check if $L^2_0-F^2_0>0$ for $x_0=x_\Delta^2$, the threshold solving \eqref{maxL2-F2}, which is the case if the cost-disadvantage $I^2/I^1$ is not too large; otherwise firm $2$ prefers to invest much later than firm $1$ and the preemption region is empty (in particular if $x_\Delta^2\geq x_F^1$, where firm $1$ would follow immediately).\footnote{\label{fn:Pempty}%
The precise condition $(I^2/I^1)^{\beta_1-1}<((1+c)^{\beta_1}-1)/(\beta_1c)$ if $c:=(D_{10}-D_{11})/(D_{11}-D_{01})\in(0,\infty)$ is obtained by plugging $x_\Delta^2=\frac{\beta_1}{\beta_1-1}\cdot\frac{I^2(r-\mu)}{(D_{10}-D_{01})^{_+}}$ (cf.\ fn.\ \ref{fn:threshold}) into the explicit functional expressions for $L^2(x_0)$ and $F^2(x_0)$, (8) and (9) in \cite{PawlinaKort06}, who obtain the same condition by a graphical argument. The condition implies $x_\Delta^2<x_F^1$. The constraint on the cost ratio strictly exceeds $1$ and is strictly increasing in $c$ to infinity by $\beta_1>1$. If $D_{10}>D_{01}\geq D_{11}$, then $x_F^1=\infty$ and the preemption region is non-empty for all $I^2\geq I^1$. Finally, if $D_{10}\leq\max\{D_{11},D_{01}\}$, then $x_\Delta^2\geq x_F^1$ and the preemption region is empty.
}

We can now characterize the equilibria of Theorem \ref{thm:SPE} for this model, which also have remarkable outcomes not captured in \cite{PawlinaKort06}. Existence is guaranteed by Proposition \ref{prop:SPE} thanks to continuity, and it suffices to solve the simpler constrained monopoly problems \eqref{MonoPreem}. By the strong Markov property, this amounts to finding the region in the state space $\R_+$ where immediate investment is optimal in the problem for $t=0$,
\begin{equation}\label{MonoPreemPK}
\sup_{\tau\leq\tau_\cP(0)\wedge\tau^2_F(0)}E\biggl[\int_{\tau}^{\infty}e^{-rs}(x_s(D_{10}-D_{00})-rI^1)\,ds\biggr].
\end{equation}
The constraint here takes the form $\min\{\tau_\cP(0),\tau^2_F(0)\}=\inf\{t\geq 0\mid x_t\in(\ubar x,\bar x)\cup[x_F^2,\infty)\}=\inf\{t\geq 0\mid x_t\in[\ubar x,\bar x]\cup[x_F^2,\infty)\}$ (\nbd{P}a.s.). Problem \eqref{MonoPreem} is then solved by investing once the state $x_t$ hits the investment region $\{x\in\R_+\mid \tau=0\text{ attains \eqref{MonoPreemPK} for }x_0=x\}$ from time $\vartheta$.

First consider a non-empty preemption region $(\ubar x,\bar x)$ that is connected to the unconstrained monopoly investment region $[x_L^1,\infty)$, as it holds for the market entry variant of the model with $D_{01}=D_{00}$, cf.\ Lemma \ref{lem:PreemMono}. Then immediate investment is optimal in \eqref{MonoPreemPK} for any state $x_0\geq\bar x\geq x_L^1$, as it is in the unconstrained problem. For states $x_0<\ubar x$ the preemption constraint in \eqref{MonoPreemPK} is a constant upper threshold, so it is optimal to wait there until $x_t$ exceeds either the constraint $\ubar x$ or the unconstrained threshold $x_L^1$, see Lemma \ref{lem:thresconstr} in Appendix \ref{app:tech}. The subgame perfect equilibrium is complete in this case: no investment for states strictly below $\min\{\ubar x,x_L^1\}$, preemptive investment in $[\ubar x,\bar x]$ as described in Section \ref{subsec:preem}, firm $1$ investing as the leader in $[x_L^1,x_F^2)\setminus[\ubar x,\bar x]$, and simultaneous investment for all states in $[x_F^2,\infty)$.

Next, if the preemption region is empty, then firm $1$ only faces the upper constraint $x_F^2$ in \eqref{MonoPreemPK}. Again by Lemma \ref{lem:thresconstr}, it is then optimal for firm $1$ to invest as soon as $x_t$ exceeds either the constraint $x_F^2$ or the unconstrained monopoly threshold $x_L^1$. Note that for the market entry variant with $D_{00}=D_{01}<D_{11}$, $x_L^1\leq x_F^1<x_F^2<\infty$. However, even if firm $1$ uses the unconstrained monopoly threshold, it is still constrained by firm $2$'s plan. Firm $1$ can only maximize the leader payoff subject to firm $2$ investing also \emph{proactively} in $[x_F^2,\infty)$.

If $D_{00}=D_{01}$, as for market entry, then preemption cannot be avoided and thus neither simultaneous investment in $[x_F^2,\infty)$ by Lemma \ref{lem:Fmaxstop}, and the equilibrium in each of the previous cases is unique. Indeed, if the preemption region is non-empty, it must contain the optimal stopping region for the continuous process $L^2_t-F^2_t$, which takes positive values only there. Then one also has to stop $L^2_t$ in that stopping region, the problem considered in Lemma \ref{lem:Lmaxstop}, because $L^2_t=(L^2_t-F^2_t)+F^2_t$ and $F^2_t$ is nonincreasing in expectation (a supermartingale) now.

So far, with $x_L^1\leq\bar x$ or $\cP=\emptyset$, investment occurs if and only if demand is high enough, i.e.\ if the state is at least $\min\{\ubar x,x_L^1\}$ or $\min\{x_F^2,x_L^1\}$. This behavior is the same for the stochastic model and its deterministic version, e.g.\ that in \cite{FudenbergTirole85} (cf.\ fn.\ \ref{fn:nest}).

\subsubsection{Preemption when demand falls}

Qualitatively different behavior can be observed in the remaining case, a monopoly threshold lying above a non-empty preemption region, $x_L^1>\bar x>\ubar x$, which requires a sufficiently high pre-investment revenue level $D_{00}>D_{01}$. Firm $1$ may then remain inactive even where it would invest immediately as follower (in states above $x_F^1$), because it has higher opportunity costs as prospective leader. This phenomenon is not addressed by \cite{PawlinaKort06}, who only consider states below $\ubar x$, where the same behavior as before holds: firm $1$ waits until $x_t$ hits the constraint $\ubar x<x_L^1$. Problem \eqref{MonoPreemPK} becomes more interesting for states in $(\bar x,x_F^2)$, where both constraints may be binding if that interval intersects the continuation region $[0,x_L^1)$ of the unconstrained problem, and behavior may be more complex. 

A lower constraint like presently $\bar x$ has a much stronger effect than any upper constraint as considered before. Two cases can be distinguished for the problem of delaying the revenue change $\pi^{L1}_t-\pi^{01}_t=e^{-rt}(x_t(D_{10}-D_{00})-rI^1)$ in $[\bar x,x_F^2]$. The easier one is that $x(D_{10}-D_{00})>rI^1$ on all of $(\bar x,x_F^2)$. Then it is optimal to invest immediately everywhere, as any delay is a loss of revenue. The more difficult case is that $x(D_{10}-D_{00})<rI^1$ near the preemption region. Firm $1$ must wait where this inequality holds, in order not to start with running losses, so one has to determine the investment region towards the upper constraint $x_F^2$. Nevertheless, it may now be optimal to invest far before the constraint is reached.

\begin{proposition}\label{prop:lowconstr}
Consider the specification \eqref{PKmodel} and suppose the corresponding preemption region $(\ubar{x},\bar x)\subset(0,x_F^1]$ from Proposition \ref{prop:PKpreemp} is non-empty. If $\bar x(D_{10}-D_{00})\geq rI^1$, then the solution of problem \eqref{MonoPreemPK} for all states $x_0$ in $(\bar x,x_F^2)$ is to invest immediately, whereas if $D_{10}-D_{00}\leq 0$, the solution is to wait until the state exits $(\bar x,x_F^2)$. 

If $0<\bar x(D_{10}-D_{00})<rI^1$, then there is a unique threshold $\hat x\in[rI^1/(D_{10}-D_{00}),x_L^1)$ solving
\begin{equation}\label{hatx}
(\beta_1-1)A(x)x^{\beta_1}+(\beta_2-1)B(x)x^{\beta_2}=I^1
\end{equation}
with
\begin{equation}\label{AB}
\begin{pmatrix} A(x) \\ B(x) \end{pmatrix}=\Bigl[\bar x^{\beta_1}x^{\beta_2}-x^{\beta_1}\bar x^{\beta_2}\Bigr]^{-1}\begin{pmatrix} x^{\beta_2} & -\bar x^{\beta_2} \\ -x^{\beta_1} & \bar x^{\beta_1} \end{pmatrix}\begin{pmatrix} \bar x\frac{D_{10}-D_{00}}{r-\mu}-I^1 \\ x\frac{D_{10}-D_{00}}{r-\mu}-I^1  \end{pmatrix}
\end{equation}
and $\beta_1>1$ and $\beta_2<0$ the roots of $\frac12\sigma^2\beta(\beta-1)+\mu\beta-r=0$, and the solution of problem \eqref{MonoPreemPK} for all states $x_0$ in $(\bar x,x_F^2)$ is to invest when $(x_t)$ exits $(\bar x,\hat x\wedge x_F^2)$.
\end{proposition}
%
%
The ``smooth-pasting'' condition that is frequently used to guess value functions only holds in the last case and only if $\hat x\leq x_F^2$. If $x_F^2(D_{10}-D_{00})\leq rI^1$, then $\hat x\geq x_F^2$ and the solution is to wait until the state exits $(\bar x,x_F^2)$. It is easy to calculate the solutions $\hat x$ of \eqref{hatx}, which are typically much lower than the upper constraint $x_F^2$ or the unconstrained threshold $x_L^1$. Thus, the risk of getting trapped at $\bar x$ by preemption induces much earlier investment, as illustrated in Section \ref{subsec:complead} below. This effect cannot be observed in the deterministic version of the model with a growing market (or falling cost).

\subsubsection{Joint investment equilibria}\label{subsec:simeql}

If $D_{00}>D_{01}$, then there are potentially many more equilibria than those from Theorem \ref{thm:SPE}, as one can now drop the premise that preemption occurs in the preemption region, and/or that simultaneous investment occurs everywhere above $x_F^2$. 

First, Proposition \ref{prop:devi} is now applied to verify equilibria of delayed joint investment, which cannot happen below $x_F^2$ for firm $2$ to invest simultaneously. The highest expected value of joint investment can be achieved by solving \eqref{sim_i}, which yields a maximal threshold, say $x_M^1$ for firm $1$. But one can also consider constrained versions of that problem, with some investment threshold ${x_J}\in[x_F^2,x_M^1]$. Joint investment triggered by ${x_J}$ is an equilibrium if firm $1$ does not want to become leader at the threshold solving problem \eqref{monoDi}, which is $\min\{x_J,x_L^1\}$ again by Lemma \ref{lem:thresconstr}. Specifically, the cost difference cannot be too large, such that firm $1$ cannot enjoy a leader's monopoly revenue for too long, which limits its leader payoff.

\begin{proposition}\label{prop:simeqlPK}
Consider the specification \eqref{PKmodel} and let $x_M^1\geq x_L^1\in[0,\infty]$ denote the threshold solving problem \eqref{sim_i} for firm $1$.\footnote{%
$x_M^1=\frac{\beta_1}{\beta_1-1}\cdot\frac{I^1(r-\mu)}{(D_{11}-D_{00})^+}$, cf.\ fn.\ \ref{fn:threshold}.
}
Suppose $x_M^1\geq x_F^2$. Then there exists a subgame perfect equilibrium of simultaneous investment triggered by the threshold ${x_J}\in[x_F^2,x_M^1]$ iff that yields firm $1$ at least the expected payoff $L^1_0$ for $x_0=x_L^1<x_F^2$, which is iff
\[
x_L^1\geq x_F^2\quad\Leftrightarrow\quad D_{10}\leq D_{00}\quad\text{or}\quad\frac{I^2}{I^1}\leq\frac{(D_{11}-D_{01})^+}{D_{10}-D_{00}}
\]
or if
\begin{equation}\label{kappa}
\biggl(\frac{I^2}{I^1}\biggr)^{\beta_1-1}\biggl[1+\biggl(\frac{x_L^1}{{x_J}}\biggr)^{\beta_1}\biggl(\beta_1-1-\frac{{x_J}}{x_L^1}\beta_1\frac{D_{11}-D_{00}}{D_{10}-D_{00}}\biggr)\biggr]\leq\beta_1\frac{D_{10}-D_{11}}{D_{10}-D_{00}}\biggl(\frac{(D_{11}-D_{01})^+}{D_{10}-D_{00}}\biggr)^{\beta_1-1}
\end{equation}
with $\beta_1>1$ from Proposition \ref{prop:lowconstr}. The LHS of \eqref{kappa} is strictly positive and strictly decreasing in ${x_J}\in[x_L^1,x_M^1]$ if $x_L^1<x_F^2$. 
\end{proposition}
%
%
Note that $x_L^1<x_F^2$ implies $D_{10}>D_{00}$. Then the second restriction on $I^2/I^1$ in the proposition is weaker than the first if setting ${x_J}=x_L^1$, and it is further relaxed if ${x_J}$ increases. If ${x_J}=x_M^1<\infty$, then \eqref{kappa} coincides with the bound on $I^2/I^1$ identified by a graphical argument in \cite{PawlinaKort06}, who impose $D_{11}>D_{00}$.\footnote{%
$x_M^1<\infty\Leftrightarrow D_{11}>D_{00}$ and then ${x_J}=x_M^1$ implies $x_J/x_L^1=(D_{10}-D_{00})/(D_{11}-D_{00})$.
} 
Proposition \ref{prop:simeqlPK} also applies for $D_{11}\leq D_{00}$, when the firms,  after both have invested, end up no better than before. Then it can still be optimal to invest at some threshold ${x_J}$ only because the other firm does, although both prefer that neither invests.

Indeed, there may be many equilibria with ``inefficient'' joint investment in states above $x_F^2$ and where the expected joint investment payoff could be improved. If $(D_{11}-D_{00})x_F^2<rI^1$, then the drift of $M^i_t$ is positive for states in the interval $[x_F^2,rI^1/(D_{11}-D_{00})^+)$, and hence it is optimal to wait in any constrained version of problem \eqref{sim_i}. Therefore one can partition the latter interval into arbitrary subintervals of alternating joint investment and waiting.

\subsubsection{Sequential investment equilibria}\label{subsec:seql}

Sequential investment without any preemption may also be an equilibrium if the preemption region is non-empty, which is a Pareto improvement compared to the equilibria of \cite{PawlinaKort06} if delayed joint investment as in Section \ref{subsec:simeql} is not feasible. Such an equilibrium can be verified by Corollary \ref{cor:seqeql} and it exists for the current specification if and only if firm $2$ does not have a strict first-mover advantage at $x_L^1$, where firm $1$ first invests.

\begin{proposition}\label{prop:seqeqlPK}
Consider the specification \eqref{PKmodel} and suppose $x_L^1<x_F^2$ (whence $D_{10}>D_{00}$). Then there exists a subgame perfect equilibrium with firm $1$ investing as soon as $x_t$ exceeds $x_L^1$ and firm $2$ planning to invest when $x_t$ exceeds $x_F^2$ iff $x_L^1\not\in(\ubar x,\bar x)$ from Proposition \ref{prop:PKpreemp}, which is iff
\[
x_L^1\geq x_F^1\quad\Leftrightarrow\quad (D_{10}-D_{00})^+\leq(D_{11}-D_{01})^+
\]
or
\begin{equation}\label{kappaseql}
(\beta_1-1)\frac{I^2}{I^1}+\biggl(\frac{I^2}{I^1}\biggr)^{1-\beta_1}\biggl(\frac{(D_{11}-D_{01})^+}{D_{10}-D_{00}}\biggr)^{\beta_1}\geq\beta_1\biggl[\frac{D_{10}-D_{01}}{D_{10}-D_{00}}-\frac{D_{10}-D_{11}}{D_{10}-D_{00}}\biggl(\frac{(D_{11}-D_{01})^+}{D_{10}-D_{00}}\biggr)^{\beta_1-1}\biggr]
\end{equation}
with $\beta_1>1$ from Proposition \ref{prop:lowconstr}. The LHS of \eqref{kappaseql} is strictly increasing in $I^2/I^1$ and the RHS is strictly positive if $x_L^1<x_F^1$. 
\end{proposition}
%
%
Finally, there may be equilibria with sequential investment as in Proposition \ref{prop:seqeqlPK} or preemption as in Proposition \ref{prop:lowconstr} where the joint investment is delayed to some threshold ${x_J}>x_F^2$, such that firm $1$ can optimize over larger intervals to become leader. This may separate the investment regions in the sequential equilibria into one where firm $1$ invests as leader and one where simultaneous investment occurs, with a gap in between. Such equilibria are more difficult to characterize explicitly. If $x_F^2$ is between two investment regions, the non-constant follower reaction prevents the simplifications used in the previous propositions.

\subsubsection{Comparison of leader investment regions}\label{subsec:complead}

To illustrate the potentially strong impact of preemption on states in $(\bar x,x_F^2)$ for varying parameter values in Figure \ref{fig:preemreg}, the model is re-parameterized as follows. First, $r$, $\mu$ and $\sigma$ determine $\beta_{1,2}$ and together with the ratio $I^1/(D_{11}-D_{01})$ also firm $1$'s follower threshold $x_F^1$, which we fix and which is an upper bound for $\bar x$. 

\begin{figure}[ht!]%
  \centering
  \subfloat{\includegraphics{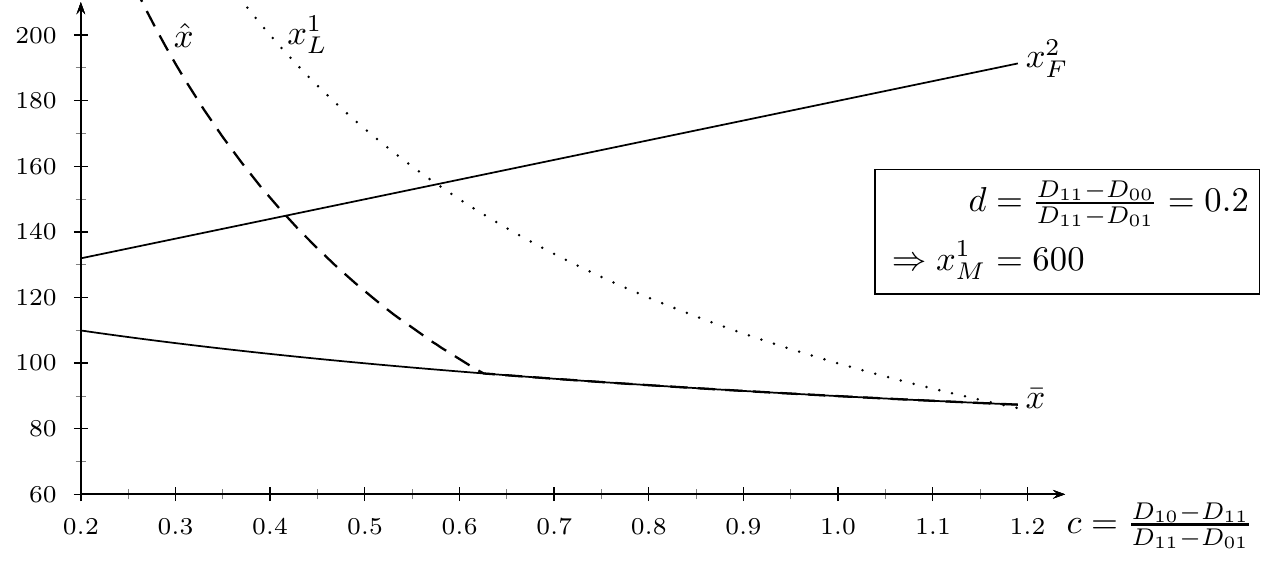}}\\ 
  \subfloat{\includegraphics{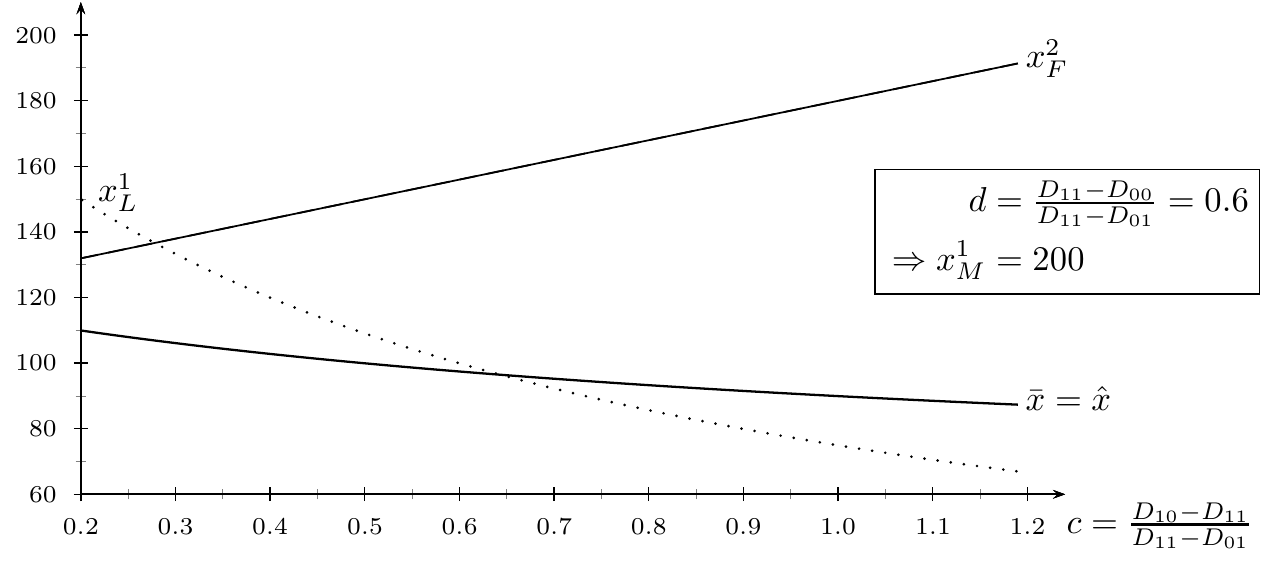}}\\
  \fbox{\small $r=0.08$, $\mu=0.02$, $\sigma=0.2$, $\frac{I^1}{D_{11}-D_{01}}=1000\Rightarrow x_F^1=120$}
  \caption{Constrained leader stopping regions.}
  \label{fig:preemreg}
\end{figure}

The distance between $\bar x$ and $x_F^2$, which is the region where firm $1$ can invest as leader, is growing in $I^2$. Indeed, $x_F^2$ obviously grows with $I^2$, and if the preemption region $(\ubar x,\bar x)$ is non-empty, it is strictly shrinking if $I^2$ grows;\footnote{%
Suppose $x_0<x_F^2$, such that firm $2$'s first-mover advantage $L^2_0-F^2_0$ is non-trivial. If $I^2$ is increased, that has two negative effects on $L^2_0-F^2_0$. First, it increases the investment cost stream $e^{-rt}rI^2$ up to firm $2$'s former follower investment time $\tau_F^2(0)$, which reduces $L^2_0$. Second, it delays $\tau_F^2(0)$. The new revenue stream difference $e^{-rt}(x_t(D_{11}-D_{01})-rI^2)$ (with increased $I^2$) between the former and the new $\tau_F^2(0)$ has non-positive expectation by optimality of the new $\tau_F^2(0)$, and thus reduces $L^2_0-F^2_0$.
}
$(\ubar x,\bar x)$ collapses when $I^2/I^1=x_F^2/x_F^1$ reaches a bound given in fn.\ \ref{fn:Pempty} in terms of $c=(D_{10}-D_{11})/(D_{11}-D_{01})$, the loss of a monopolist relative to the gain of the follower when the latter invests. We pick those limit values for $I^2$ and $x_F^2$ for simplicity, thus making both functions of $c$, although then just $\ubar x=\bar x=x_\Delta^2$, the threshold solving \eqref{maxL2-F2}. Now $c$ also determines $\bar x$ by $x_\Delta^2=x_F^2/(1+c)$. 

Equation \eqref{hatx} for $\hat x$ can be reduced to the parameters $\beta_{1,2}$ and $x_L^1$, the unconstrained monopoly threshold, which is an upper bound on $\hat x$ and itself satisfies $x_L^1=x_F^1/(c+d)$ with $d:=(D_{11}-D_{00})/(D_{11}-D_{01})$. The latter ratio comes close to $1$ if the leader's investment has not much influence on the follower's revenue, like in a market entry situation; it becomes small when the leader steals considerable business from the follower, like by a drastic innovation. $d$ also controls the best simultaneous investment threshold by $x_M^1=x_F^1/d$.

In the equilibria from Theorem \ref{thm:SPE}, firm $1$ can freely decide when to invest in the interval $(\bar x,x_F^2)$. Without the threat of preemption, it would not invest below $\min\{x_L^1,x_F^2\}$. However, given the threat of preemption, firm $1$ already invests when the state exceeds $\hat x$, which may be much earlier as Figure \ref{fig:preemreg} shows. In the upper panel with a low value of $d$, the threat of preemption strongly matters for $c\geq 0.45$. Firm $1$ never chooses to wait at all in the lower panel with a moderate value of $d$. Joint investment at $x_M^1$ is an equilibrium avoiding preemption if $x_L^1\geq x_F^2$; it is not an equilibrium for $d=0.6$ and $c\geq 0.45$.

\subsection{Strategic real estate development with construction time}\label{subsec:Gren96}

Similar reasoning as before shows on the one hand that equilibria discussed in \cite{Grenadier96} only exist under certain parameter restrictions and on the other hand  that there exist additional equilibria that are Pareto improvements.

\cite{Grenadier96} models a real option game between two symmetric real estate owners, who may each invest to redevelop their property in order to earn higher rents. His model needs a slight translation to fit into the current framework, as it includes a delay of construction: if an owner invests, it takes $\delta\geq 0$ time units until the new building yields any revenues. Before investment by any owner, both earn the deterministic rent $R\geq 0$. Investment at cost $I>0$ terminates that rent, reduces the rent of the opponent to $(1-\gamma)R$ with $\gamma\in[0,1]$ and initiates new own rent $D_1x_t$ after the delay $\delta$. $(x_t)$ is a geometric Brownian motion as in \eqref{GBM}. Once both new buildings are completed, each owner earns the rent $D_2x_t$ with $0<D_2\leq D_1$.

Grenadier's model is strategically equivalent to specifying
\begin{equation*}
\begin{aligned}
\pi^{0i}_t&=e^{-rt}R, & \pi^{Li}_t&=e^{-rt}(D_1e^{-(r-\mu)\delta}x_t-rI), \\ 
\pi^{Fi}_t&=e^{-rt}(1-\gamma)R,\quad & \pi^{Bi}_t&=e^{-rt}(D_2e^{-(r-\mu)\delta}x_t-rI)
\end{aligned}
\end{equation*}
in the general framework. The equilibria proposed in \cite{Grenadier96} are justified by the insufficient argument that waiting is optimal if the current follower payoff exceeds the current leader payoff. Nevertheless there exists a subgame perfect equilibrium as in Theorem \ref{thm:SPE} by symmetry; it can be characterized as follows. The follower problems \eqref{F} are again solved by investing once $x_t$ exceeds a threshold $x_F>0$, whence simultaneous investment is an equilibrium for all states $x_\vartheta\geq x_F$.\footnote{%
$x_F=\frac{\beta_1}{\beta_1-1}\cdot e^{(r-\mu)\delta}(I+(1-\gamma)R/r)(r-\mu)/D_2$ with $\beta_1>1$ from fn.\ \ref{fn:threshold}.
} 
Problem \eqref{maxL2-F2} is solved by a threshold $x_\Delta=x_FD_2/D_1$ and the preemption region $\cP$ is in fact non-empty if and only if $D_2<D_1$. $\cP$ can be represented by an interval $(\ubar x,\bar x)$ of the state space by the same arguments as in the proof of Proposition \ref{prop:PKpreemp}, where now $\bar x=x_F$.


\subsubsection{Qualification of further equilibria}

Depending on the parameter values, there may be other equilibria with delayed simultaneous investment and/or no preemption.  Let $x_L$ denote the threshold solving the present instance of the unconstrained monopoly problem \eqref{Mono^i}.\footnote{%
$x_L=\frac{\beta_1}{\beta_1-1}e^{(r-\mu)\delta}(I+R/r)(r-\mu)/D_1$ with $\beta_1>1$ from fn.\ \ref{fn:threshold}. This should not be confused with $X_L$ in \cite{Grenadier96}, which corresponds to the present $\ubar x$.
} 
For states above $\bar x=x_F$, any investment will be simultaneous. Contrarily to the claim made in \cite{Grenadier96}, simultaneous investment cannot be delayed past the threshold $x_M=x_LD_1/D_2\geq x_F$ solving problem \eqref{sim_i}. Indeed, in any equilibrium with preemption in $\cP$, by symmetry both firms get at most the follower payoff at the time of investment. The same holds for any equilibrium with only joint investment. In either case investment must occur as soon as the state exceeds $x_M$, because any delay would be a loss by Lemma \ref{lem:Fmaxstop}. 

With preemption occuring in $\cP$, one can only consider delaying simultaneous investment in the interval $[\bar x,x_M]$, i.e.\ delaying the revenue change $\pi^{Bi}_t-\pi^{0i}_t=e^{-rt}(D_2e^{-(r-\mu)\delta}x_t-rI-R)$. This problem has the same form as the one with two-sided constraint considered in Proposition \ref{prop:lowconstr} (recall also the illustration in Section \ref{subsec:complead}), with $D_2e^{-(r-\mu)\delta}$ replacing $D_{10}-D_{00}$, $I+R/r$ replacing $I^1$ and $x_M$ replacing $x_F^2$. Thus, given now $\bar x=x_F$, if $D_2e^{-(r-\mu)\delta}x_F\geq rI+R$, which means if
\begin{equation}\label{gammaGren}
\gamma\leq\biggl(\frac{rI}{R}+1\biggr)\biggl(1-\frac{\beta_1-1}{\beta_1(r-\mu)}\biggr),
\end{equation}
then investment cannot be delayed at all for states above $x_F$, which is not recognized in \cite{Grenadier96}. In this case the preemption region extends to such high states that any foregone revenue above it is a loss. Note that the RHS of \eqref{gammaGren} is strictly positive.

Only if \eqref{gammaGren} fails, there will exist a solution $\hat x\in[(rI+R)e^{(r-\mu)\delta}/D_2,x_M)$ to the current version of \eqref{hatx}, such that investment can be held back in $(x_F,\hat x)$. Only then the phenomenon discussed extensively in Section V of \cite{Grenadier96} can arise, that preemption occurs when demand \emph{falls} to $x_F$.
 
However, if $\gamma$ is sufficiently large to violate \eqref{gammaGren}, then delayed joint investment may be attractive enough to avoid preemption altogether, which will be a Pareto improvement w.r.t.\ \cite{Grenadier96}. By the same arguments as for Proposition \ref{prop:simeqlPK}, preemption can be avoided in an equilibrium of joint investment with the threshold $x_M\geq x_F$ if and only if that yields firm $1$ at least the expected payoff $L^1_0$ for $x_0=x_L<x_F$, which is if and only if
\begin{equation*}
x_L\geq x_F \quad\Leftrightarrow\quad \gamma\geq\biggl(\frac{rI}{R}+1\biggr)\biggl(1-\frac{D_2}{D_1}\biggr)
\end{equation*}
or if
\[
\gamma\geq\biggl(\frac{rI}{R}+1\biggr)\biggl(1-D_2\biggl(\beta_1\frac{D_1-D_2}{D_1^{\beta_1}-D_2^{\beta_1}}\biggr)^{\frac{1}{\beta_1-1}}\biggr)
\]
with $\beta_1>1$ from fn.\ \ref{fn:threshold}. The last restriction on $\gamma$ is indeed weaker than the previous one.

\section{Conclusion}\label{sec:conc}
 
The equilibrium analysis of the general model in Section \ref{sec:eql} was based directly on its primitives and not on derived analytic properties of value functions, as it frequently happens in the growing literature on real option games. By that more general perspective, there is on the one hand less risk to neglect any verification problems for equilibria and on the other hand a more detailed view of their economic structure. For models that satisfy the general assumptions made here, the number of equilibrium verification problems has been reduced considerably by elementary economic arguments and it remains to solve a single class of optimal stopping problems for one firm. Theorem \ref{thm:SPE} applies to many more examples from the literature than the ones revisited in Section \ref{sec:example} (e.g.\ to those listed in the Introduction). The presented applications, which have quite distinctive economic properties, show how the general results act in typical state-space models. By the more complete approach, some neglected equilibrium behavior has been identified that qualitatively distinguishes stochastic from deterministic models. In particular two-sided constraints induce feedback effects when the state evolves randomly. The arguments developed for the identification of additional equilibria that may be Pareto improvements also generalize to other models, e.g.\ for the source of uncertainty.

Thus the general perspective taken here provides a foundation for a more complete analysis of models of preemptive investment that fit into the framework and a guideline for the analysis of further models that do not satisfy the revenue orders assumed here.

\appendix {

\section{Appendix}\label{app:add}

\subsection{Characterizing the preemption region}\label{app:locpreem}

To see if the preemption region is empty, it suffices to consider stopping times that are optimal for some simple stopping problems. They are the solutions of firm $i$'s permanent monopoly problem \eqref{Mono^i} if $\pi^{0i}_\cdot\equiv\pi^{Fi}_\cdot$ (like in a market entry model).

\begin{lemma}\label{lem:PreemMono}
For any $\vartheta\in\T$, $L^2_\vartheta>F^2_\vartheta$ only if $E\bigl[L^2_{\tau_\Delta^i}-F^2_{\tau_\Delta^i}\bigv\F_\vartheta\bigr]>0$ for all times $\tau_\Delta^i\in\T$ attaining
\begin{equation}\label{maxL2-F2}
\esssup_{\tau\geq\vartheta}E\biggl[\int_0^{\tau}\pi^{Fi}_s\,ds+\int_{\tau}^{\infty}\pi^{Li}_s\,ds\biggv\F_\vartheta\biggr]
\end{equation}
for some $i\in\{1,2\}$. Where $\tau_\Delta^2=\vartheta$ attains \eqref{maxL2-F2} for $i=2$, there $L^2_\vartheta-F^2_\vartheta\geq E\bigl[L^2_\tau-F^2_\tau\bigv\F_\vartheta\bigr]$ for all $\tau\in[\vartheta,\tau_F^1(\vartheta)]$.
\end{lemma}

Lemma \ref{lem:PreemMono} rests on the fact that for any $\tau\in[\vartheta,\tau_F^2(\vartheta)]$, the difference between $L^2_\vartheta$ and $F^2_\vartheta$ on $[\vartheta,\tau]$ is that between the monopoly or duopoly revenue and the laggard's revenue, i.e.\ at most $\pi^{L2}_\cdot-\pi^{F2}_\cdot$. That difference is nonpositive in expectation up to any solution of \eqref{maxL2-F2}, where indeed $\tau_\Delta^2\leq\tau_F^2(\vartheta)$ by $\pi^{L2}_\cdot\geq\pi^{B2}_\cdot$. Further, the revenue difference between $L^2_\vartheta$ and $F^2_\vartheta$ on $[\tau_\Delta^2,\infty)$ is at most that between $L^2_{\tau_\Delta^2}$ and $F^2_{\tau_\Delta^2}$, because firm $2$'s follower reaction remains the same and, by becoming leader later, firm $2$ receives the monopoly revenue at least until the same time.

For state-space models like in Section \ref{sec:example}, we get the following characterization. First, as noted in Subsection \ref{subsec:preem}, a follower threshold for either firm $i$, say $x_F^i\in\R$, is never contained in the preemption region,\footnote{%
Here ``the preemption region'' refers to an area in the same state space in which the thresholds are defined, which is of course an abuse of terminology regarding the previous definition of $\cP$.
} 
not even in its closure if investment at $x_F^1$ is not optimal for firm $2$. As $L^2_\cdot\leq F^2_\cdot$ for all states above such $x_F^i$, the latter must lie above any non-empty preemption region. Second, by Lemma \ref{lem:PreemMono}, any non-empty preemption region must intersect the stopping regions from \eqref{maxL2-F2} for both $i=1,2$; a threshold solving that problem, say $x_\Delta^i\in\R$, cannot lie above the preemption region. In particular, if $x_\Delta^2\geq x_F^1$, the preemption region must be empty. Third, if firm $2$ has no first-mover advantage at $x_\Delta^2$, then it has none at any value that the state will attain before crossing $x_F^1$. Thus, if the state, starting from some $x_\Delta^2<x_F^1$, will attain any intermediate value before reaching $x_F^1$, then it suffices to check whether there is a first-mover advantage for firm $2$ at $x_\Delta^2$; otherwise the preemption region is empty, because $x_\Delta^2$ cannot lie above it.

\begin{proof}[{\bf Proof of Lemma \ref{lem:PreemMono}}]
First note that there are solutions $\tau_\Delta^i\leq\tau^i_F(\vartheta)\leq\tau^2_F(\vartheta)$ to \eqref{maxL2-F2} for $i=1,2$ as the respective process to be stopped is continuous and integrable. The estimate follows from the assumption $\pi^{Li}_\cdot-\pi^{Fi}_\cdot\geq\pi^{Bi}_\cdot-\pi^{Fi}_\cdot$, cf.\ the proof of Lemma \ref{lem:tau^1_F<tau^2_F}. 

By the optimality of $\tau_\Delta^i$ in \eqref{maxL2-F2}, $E\bigl[\int_{\vartheta}^{\tau_\Delta^i}(\pi^{Li}_s-\pi^{Fi}_s)\,ds\bigv\F_\vartheta\bigr]\leq 0$. Therefore, as $\pi^{L2}_\cdot-\pi^{F2}_\cdot\leq\pi^{Li}_\cdot-\pi^{Fi}_\cdot$, \eqref{L^2-F^2} can only be strictly positive if 
\[
E\biggl[\int_{\tau_\Delta^i}^{\tau^1_F(\vartheta)}(\pi^{L2}_s-\pi^{F2}_s)\,ds+\int_{\tau^1_F(\vartheta)}^{\tau^2_F(\vartheta)}(\pi^{B2}_s-\pi^{F2}_s)\,ds\biggv\F_{\vartheta}\biggr]>0
\]
(which can in fact only be the case if $P[\tau_\Delta^i<\tau^1_F(\vartheta)]>0$), and which implies
\begin{align*}
&E\Bigl[L^2_{\tau_\Delta^i}-F^2_{\tau_\Delta^i}\Bigv\F_\vartheta\Bigr]=E\biggl[\int_{\tau_\Delta^i}^{\tau^1_F(\tau_\Delta^i)}(\pi^{L2}_s-\pi^{F2}_s)\,ds+\int_{\tau^1_F(\tau_\Delta^i)}^{\tau^2_F(\vartheta)}(\pi^{B2}_s-\pi^{F2}_s)\,ds\biggv\F_{\vartheta}\biggr]>0
\end{align*}
as $\tau^1_F(\tau_\Delta^i)\geq\tau^1_F(\vartheta)$, $\tau^2_F(\tau_\Delta^i)=\tau^2_F(\vartheta)$ and $\pi^{L2}_\cdot\geq\pi^{B2}_\cdot$. 

For all stopping times $\tau\in[\vartheta,\tau^1_F(\vartheta)]$, indeed $\tau^i_F(\tau)=\tau^i_F(\vartheta)$, $i=1,2$, and thus $L^2_\vartheta-F^2_\vartheta-E\bigl[L^2_\tau-F^2_\tau\bigv\F_\vartheta\bigr]=E\bigl[\int_{\vartheta}^{\tau}(\pi^{L2}_s-\pi^{F2}_s)\,ds\bigv\F_\vartheta\bigr]\geq 0$ if $\tau_\Delta^2=\vartheta$ attains \eqref{maxL2-F2}.
%
\end{proof}

\subsection{Verification of equilibria without preemption}\label{app:othereql}

The following proposition helps to reduce the search for times at which firm $i$ may want to preempt firm $j$ and thus to verify a best reply $\tau_*^i\geq\tau_*^j$. It avoids to maximize the leader payoff directly, which is a complex problem due to the follower reaction. Applied to state-space models, it may suffice to consider deviations at a single threshold.

\begin{proposition}\label{prop:devi}
Consider any given $\vartheta\in\T$ and $i,j\in\{1,2\}$, $i\not=j$. If firm $j$ plans to invest at the stopping time $\tau_*^j\geq\vartheta$, then $\tau_*^i\geq\tau_*^j$ is a best reply for firm $i$ if $F^i_{\tau_*^j}=M^i_{\tau_*^j}$ on $\{\tau_*^i=\tau_*^j\}$ and
\begin{enumerate}
\item\label{maxM_D}
$E\bigl[F^i_{\tau_*^j}\bigv\F_{\vartheta}\bigr]\geq\esssup_{\tau\in[\vartheta,\tau_*^j]}E\bigl[M^i_\tau\bigv\F_{\vartheta}\bigr]$ and
\item\label{maxL_D}
for each stopping time $\vartheta'\geq\vartheta$, on $\{\vartheta'<\tau_*^j\}$ one of the solutions $\tau_D^i(\vartheta')\in\T$ of the problem
\begin{equation}\label{monoDi}
\esssup_{\tau\in[\vartheta',\tau_*^j\vee\vartheta']}E\biggl[\int_0^\tau\pi^{0i}_s\,ds+\int_{\tau}^{\infty}\pi^{Li}_s\,ds\biggv\F_{\vartheta'}\biggr]
\end{equation}
satisfies either $\tau_D^i(\vartheta')\geq\tau_F^j(\vartheta')$ or $L^i_{\tau_D^i(\vartheta')}\leq E\bigl[F^i_{\tau_*^j}\bigv\F_{\tau_D^i(\vartheta')}\bigr]$.
\end{enumerate}

Where $\vartheta'$ attains \eqref{monoDi}, it holds that $L^i_{\vartheta'}-E\bigl[F^i_{\tau_*^j}\bigv\F_{\vartheta'}\bigr]\geq E\bigl[L^i_\tau-F^i_{\tau_*^j}\bigv\F_{\vartheta'}\bigr]$ for all stopping times $\tau\in[\vartheta',\tau_F^j(\vartheta')]$.

Further, if $\pi^{L1}_\cdot-\pi^{01}_\cdot\geq\pi^{L2}_\cdot-\pi^{02}_\cdot$, $\pi^{B1}_\cdot-\pi^{01}_\cdot\geq\pi^{B2}_\cdot-\pi^{02}_\cdot$, $F^2_{\tau_*^2}=M^2_{\tau_*^2}$ and \ref{maxM_D}, \ref{maxL_D} hold for $i=1$, then $\tau_*^1=\tau_*^2$ are mutual best replies.
\end{proposition}

Condition \ref{maxM_D} is obviously also necessary, as the terminal payoff is at most $F^i_{\tau_*^j}$ (without preemption modeled as in Section \ref{subsec:preem}) and $L^i_\cdot\geq M^i_\cdot$. Condition \ref{maxL_D} says that it suffices to check for deviations by firm $i$ at solutions $\tau_D^i(\vartheta')<\tau_F^j(\vartheta')$ of \eqref{monoDi}, so there is nothing to check where $\vartheta'=\tau_F^j(\vartheta')$. The subsequent sentence implies that for threshold-type models, it is typically enough to consider $\vartheta'=\tau_D^i(\vartheta)$: If firm $i$ does not want to become leader there, it does not at any value that the state process will attain before crossing firm $j$'s follower threshold that determines $\tau_F^j(\vartheta)$. For states above that threshold, no deviations need to be considered.

Proposition \ref{prop:devi} can be applied in particular to equilibria of joint investment at some time $\tau_J=\tau_*^1=\tau_*^2\geq\vartheta$. Then on the one hand $F^2_{\tau_J}=M^2_{\tau_J}$ is necessary, which automatically implies $F^1_{\tau_J}=M^1_{\tau_J}$ by Lemma \ref{lem:tau^1_F<tau^2_F}. On the other hand, \ref{maxM_D} is then the clearly necessary condition that $\tau_J$ must be an (at least constrained) optimal time for maximizing the expected joint investment payoff $E\bigl[M^i_{\tau_J}\bigv\F_\vartheta\bigr]$ as considered in Lemma \ref{lem:Fmaxstop}. Given such $\tau_J$, an equilibrium can be verified by condition \ref{maxL_D}, where it suffices to consider firm $1$ if the additional revenue order holds.

\begin{proof}[{\bf Proof of Proposition \ref{prop:devi}}]
Given $\tau_*^j\geq\vartheta$, firm $i$'s expected payoff from any stopping time $\tau^i\geq\vartheta$ is $E\bigl[L^i_{\tau^i}\indi{\tau^i<\tau_*^j}+M^i_{\tau^i}\indi{\tau^i=\tau_*^j}+F^i_{\tau_*^j}\indi{\tau^i>\tau_*^j}\bigv\F_{\vartheta}\bigr]\leq E\bigl[L^i_{\tau^i}\indi{\tau^i<\tau_*^j}+F^i_{\tau_*^j}\indi{\tau_*^i\geq\tau_*^j}\bigv\F_{\vartheta}\bigr]$. The latter is attainable by the stopping time $\tau^i\indi{\tau^i<\tau_*^j}+\infty\indi{\tau_*^i\geq\tau_*^j}$, so $\tau_*^i$ is a best reply to $\tau_*^j$ iff $F^i_{\tau_*^j}=M^i_{\tau_*^j}$ on $\{\tau_*^i=\tau_*^j\}$ and $\tau=\tau_*^j$ attains
\begin{align*}
&\esssup_{\vartheta\leq\tau\leq\tau_*^j}E\Bigl[L^i_\tau\indi{\tau<\tau_*^j}+F^i_{\tau_*^j}\indi{\tau\geq\tau_*^j}\Bigv\F_\vartheta\Bigr].
\end{align*}
By iterated expectations, this is equivalent to $L^i_{\vartheta'}-E\bigl[F^i_{\tau_*^j}\bigv\F_{\vartheta'}\bigr]\leq 0$ on $\{\vartheta'<\tau_*^j\}$ for all stopping times $\vartheta'\geq\vartheta$. To establish the latter under conditions \ref{maxM_D} and \ref{maxL_D}, fix arbitrary $\vartheta'\geq\vartheta$ and let $\tau_D^i(\vartheta')\in\T$ attain \eqref{monoDi} (such $\tau_D^i(\vartheta')$ exists by continuity and integrability of the process to be stopped), whence $E\bigl[\int_{\vartheta'}^{\tau_D^i(\vartheta')}(\pi^{Li}_s-\pi^{0i}_s)\,ds\bigv\F_{\vartheta'}\bigr]\leq 0$. On $\{\vartheta'<\tau_*^j\}$ we then have
\begin{align}\label{devDi}
L^i_{\vartheta'}-E\Bigl[M^i_{\tau_*^j}\Bigv\F_{\vartheta'}\Bigr]&=E\biggl[\int_{\vartheta'}^{\tau_F^j(\vartheta')}(\pi^{Li}_s-\pi^{0i}_s)\,ds+\int_{\tau_F^j(\vartheta')}^{\tau_*^j}(\pi^{Bi}_s-\pi^{0i}_s)\,ds\biggv\F_{\vartheta'}\biggr] \\
&\leq E\biggl[\int_{\vartheta'}^{\tau_F^j(\vartheta')\vee\tau_D^i(\vartheta')}(\pi^{Li}_s-\pi^{0i}_s)\,ds+\int_{\tau_F^j(\vartheta')\vee\tau_D^i(\vartheta')}^{\tau_*^j}(\pi^{Bi}_s-\pi^{0i}_s)\,ds\biggv\F_{\vartheta'}\biggr] \nonumber\\
&\leq E\biggl[\int_{\tau_D^i(\vartheta')}^{\tau_F^j(\vartheta')\vee\tau_D^i(\vartheta')}(\pi^{Li}_s-\pi^{0i}_s)\,ds+\int_{\tau_F^j(\vartheta')\vee\tau_D^i(\vartheta')}^{\tau_*^j}(\pi^{Bi}_s-\pi^{0i}_s)\,ds\biggv\F_{\vartheta'}\biggr] \nonumber\\
&=E\Bigl[\indi{\tau_D^i(\vartheta')<\tau_F^j(\vartheta')}\Bigl(L^i_{\tau_D^i(\vartheta')}-M^i_{\tau_*^j}\Bigr) \nonumber\\
&\quad\;+\indi{\tau_D^i(\vartheta')\geq\tau_F^j(\vartheta')}\Bigl(M^i_{\tau_D^i(\vartheta')}-M^i_{\tau_*^j}\Bigr)\Bigv\F_{\vartheta'}\Bigr]. \nonumber
\end{align}
The first equality uses the convention $\int_b^a\cdot\,ds=-\int_a^b\cdot\,ds$ for $a<b$. The first inequality is due to $\pi^{Li}_\cdot\geq\pi^{Bi}_\cdot$ and the second due to the optimality of $\tau_D^i(\vartheta')$. The last equality is analogous to the first, using iterated expectations and $\tau_D^i(\vartheta')<\tau_F^j(\vartheta')\Rightarrow\tau_F^j(\tau_D^i(\vartheta'))=\tau_F^j(\vartheta')$. After replacing $M^i_{\tau_*^j}$ by $F^i_{\tau_*^j}$ in the first and last terms of \eqref{devDi}, conditions \ref{maxM_D} and \ref{maxL_D} make the last nonpositive (taking iterated expectations at $\tau_D^i(\vartheta')$), and thus also $L^i_{\vartheta'}-E\bigl[F^i_{\tau_*^j}\bigv\F_{\vartheta'}\bigr]\leq 0$.

To prove the next claim, note that for any stopping time $\tau\in[\vartheta',\tau_F^j(\vartheta')]$ we have $\tau_F^j(\tau)=\tau_F^j(\vartheta')$ and thus $L^i_{\vartheta'}-E\bigl[L^i_{\tau}\bigv\F_{\vartheta'}\bigr]=E\bigl[\int_{\vartheta'}^{\tau}(\pi^{Li}_s-\pi^{0i}_s)\,ds\bigv\F_{\vartheta'}\bigr]\geq 0$ where $\vartheta'$ attains \eqref{monoDi}.

For the final claim consider any stopping time $\tau_*^2\geq\vartheta$ such that $F^2_{\tau_*^2}=M^2_{\tau_*^2}$; then also $F^1_{\tau_*^2}=M^1_{\tau_*^2}$ by Lemma \ref{lem:tau^1_F<tau^2_F}. Suppose further that \ref{maxM_D}, \ref{maxL_D} hold for $i=1$, so $\tau_*^1=\tau_*^2$ is a best reply for firm $1$. To prove that $\tau_*^2$ is also a best reply for firm $2$ to $\tau_*^1=\tau_*^2$ if $\pi^{L1}_\cdot-\pi^{01}_\cdot\geq\pi^{L2}_\cdot-\pi^{02}_\cdot$ and $\pi^{B1}_\cdot-\pi^{01}_\cdot\geq\pi^{B2}_\cdot-\pi^{02}_\cdot$, we show that then \eqref{devDi} is not greater for $i=2$ than for $i=1$. Therefore note that for each $i=1,2$, $F^i_{\tau_*^2}=M^i_{\tau_*^2}$ implies $E\bigl[\indi{A}\int_{\tau_*^2}^{\tau_F^i(\vartheta')}(\pi^{Bi}_s-\pi^{Fi}_s)\,ds\bigv\F_{\vartheta'}\bigr]=0$ for any set $A\subset\{\tau_F^i(\vartheta')\geq\tau_*^2\}$ (taking iterated expectations at $\tau_*^2$), in particular for $A=\{\tau_F^1(\vartheta')>\tau_*^2\}$ as $\tau_F^2(\vartheta')\geq\tau_F^1(\vartheta')$. Further, $E\bigl[\indi{\tau_F^1(\vartheta')>\tau_*^2}\int_{\tau_F^1(\vartheta')}^{\tau_F^2(\vartheta')}(\pi^{B2}_s-\pi^{F2}_s)\,ds\bigv\F_{\vartheta'}\bigr]\leq 0$ by optimality of $\tau_F^2(\vartheta')$ (and iterated expectations at $\tau_F^1(\vartheta')$), so $E\bigl[\indi{\tau_F^1(\vartheta')>\tau_*^2}\int_{\tau_*^2}^{\tau_F^1(\vartheta')}(\pi^{B2}_s-\pi^{F2}_s)\,ds\bigv\F_{\vartheta'}\bigr]\geq 0$.

Now, rewriting \eqref{devDi} for $i=2$, we obtain
\begin{align}\label{devD2<devD1}
E\biggl[\int_{\vartheta'}^{\tau_F^1(\vartheta')\wedge\tau_*^2}(\pi^{L2}_s-\pi^{02}_s)\,ds+\indi{\tau_F^1(\vartheta')\leq\tau_*^2}&\int_{\tau_F^1(\vartheta')}^{\tau_*^2}(\pi^{B2}_s-\pi^{02}_s)\,ds \nonumber\\
+\indi{\tau_F^1(\vartheta')>\tau_*^2}&\int_{\tau_*^2}^{\tau_F^1(\vartheta')}(\pi^{L2}_s-\pi^{B2}_s)\,ds\biggv\F_{\vartheta'}\biggr] \nonumber\\
\leq E\biggl[\int_{\vartheta'}^{\tau_F^1(\vartheta')\wedge\tau_*^2}(\pi^{L1}_s-\pi^{01}_s)\,ds+\indi{\tau_F^1(\vartheta')\leq\tau_*^2}&\int_{\tau_F^1(\vartheta')}^{\tau_*^2}(\pi^{B1}_s-\pi^{01}_s)\,ds \nonumber\\
+\indi{\tau_F^1(\vartheta')>\tau_*^2}&\int_{\tau_*^2}^{\tau_F^1(\vartheta')}(\pi^{L2}_s-\pi^{F2}_s)\,ds\biggv\F_{\vartheta'}\biggr] \nonumber\\
\leq E\biggl[\int_{\vartheta'}^{\tau_F^1(\vartheta')\wedge\tau_*^2}(\pi^{L1}_s-\pi^{01}_s)\,ds+\indi{\tau_F^1(\vartheta')\leq\tau_*^2}&\int_{\tau_F^1(\vartheta')}^{\tau_*^2}(\pi^{B1}_s-\pi^{01}_s)\,ds \nonumber\\
+\indi{\tau_F^1(\vartheta')>\tau_*^2}&\int_{\tau_*^2}^{\tau_F^1(\vartheta')}(\pi^{L1}_s-\pi^{F1}_s)\,ds \nonumber\\
+{}&\int_{\tau_F^1(\vartheta')}^{\tau_F^2(\vartheta')}(\pi^{L1}_s-\pi^{B1}_s)\,ds\biggv\F_{\vartheta'}\biggr]
\end{align}
The last inequality uses the assumption $\pi^{L1}_\cdot-\pi^{F1}_\cdot\geq\pi^{L2}_\cdot-\pi^{F2}_\cdot$ as well as $\tau_F^1(\vartheta')\leq\tau_F^2(\vartheta')$ and $\pi^{L1}_\cdot\geq\pi^{B1}_\cdot$. Rearranging \eqref{devD2<devD1} using $E\bigl[\indi{\tau_F^1(\vartheta')>\tau_*^2}\int_{\tau_*^2}^{\tau_F^i(\vartheta')}(\pi^{Bi}_s-\pi^{Fi}_s)\,ds\bigv\F_\vartheta\bigr]=0$ yields \eqref{devDi} for $i=1$.
\end{proof}

Proposition \ref{prop:devi} simplifies as follows for sequential investment.

\begin{corollary}\label{cor:seqeql}
Consider any $\vartheta\in\T$ and let $\tau_S\in\T$ solve \eqref{maxL_S}. Then it is an equilibrium in the subgame beginning at $\vartheta$ that firm $1$ plans to invest at $\tau_*^1=\tau_S$ and firm $2$ at $\tau_*^2=\tau_F^2(\vartheta)$ if condition \ref{maxL_D} of Proposition \ref{prop:devi} is satisfied for firm $i=2$.

Further, if $\pi^{L1}_\cdot-\pi^{01}_\cdot\geq\pi^{L2}_\cdot-\pi^{02}_\cdot$, then $\tau_D^2(\vartheta')=\tau_S$ attains \eqref{monoDi} where $\vartheta'\leq\tau_*^1=\tau_S$.
\end{corollary}

Note that in the setting of Corollary \ref{cor:seqeql}, it suffices for condition \ref{maxL_D} of Proposition \ref{prop:devi} to hold that firm $2$ does not have a local first-mover advantage where $\tau_D^2(\vartheta')<\tau_F^1(\vartheta')$ attains \eqref{monoDi}, as $(F^2_t)$ is a submartingale on $[\vartheta',\tau_F^2(\vartheta')]$. Under the additional revenue order in the corollary, this simply amounts to $\tau_S$ not being in the preemption region $\cP$.

\begin{proof}[{\bf Proof of Corollary \ref{cor:seqeql}}]
We only need to verify optimality for firm $i=2$ by applying Proposition \ref{prop:devi} with $\tau_*^1=\tau_S\leq\tau_F^2(\vartheta)=\tau_*^2$. Then indeed $F^2_{\tau_*^2}=M^2_{\tau_*^2}$. Further, condition \ref{maxM_D} is satisfied as $M^2_\cdot\leq F^2_\cdot$ and $(F^2_t)$ is a submartingale on $[\vartheta,\tau_F^2(\vartheta)]$ by $\pi^{F2}_\cdot\leq\pi^{02}_\cdot$. Hence $\tau_*^2$ is optimal if the remaining condition \ref{maxL_D} is satisfied.

For the second claim note that if $\pi^{L1}_\cdot-\pi^{01}_\cdot\geq\pi^{L2}_\cdot-\pi^{02}_\cdot$, then $E\bigl[\int_{\tau}^{\tau_S}(\pi^{L2}_s-\pi^{02}_s)\,ds\bigv\F_{\tau}\bigr]\leq E\bigl[\int_{\tau}^{\tau_S}(\pi^{L1}_s-\pi^{01}_s)\,ds\bigv\F_{\tau}\bigr]\leq 0$ for any stopping time $\tau\in[\vartheta,\tau_S]$ by the optimality of $\tau_S$, cf.\ Lemma \ref{lem:stoptau^1_F}, and thus $\tau_D^2(\vartheta')=\tau_S\vee\vartheta'$ attains the current instance of \eqref{monoDi}.
\end{proof}

\subsection{Technical results}\label{app:tech}

\begin{lemma}\label{lem:LFreg}
In the setting of Section \ref{sec:model}, consider four processes $(\pi^m_t)\in L^1(dt\otimes P)$, $m=0,L,F,B$, such that each process $(\int_0^t\pi^m_s\,ds)$ is adapted, and let $\bigl\{\tau_O(\tau),\tau\in\T\bigr\}$ be a family of stopping times satisfying $\tau\leq\tau_O(\tau)\leq\tau_O(\tau')$ a.s.\ for all $\tau,\tau'\in\T$ with $\tau\leq\tau'$ a.s. Then there exist optional processes $(L_t)$ and $(F_t)$ that are of class {\rm (D)} and which satisfy
\begin{equation*}
L_\tau=L(\tau):=\int_0^\tau \pi^0_s\,ds+E\biggl[\int_\tau^{\tau_O(\tau)}\pi^L_s\,ds+\int_{\tau_O(\tau)}^\infty\pi^B_s\,ds\biggv\F_\tau\biggr]
\end{equation*}
and
\begin{equation*}
F_\tau=F(\tau):=\int_0^\tau \pi^0_s\,ds+\esssup_{\tau'\geq \tau}E\biggl[\int_\tau^{\tau'}\pi^F_s\,ds+\int_{\tau'}^\infty\pi^B_s\,ds\biggv\F_\tau\biggr]
\end{equation*}
a.s.\ for every $\tau\in\T$. In particular, the process $(F_t)$ can be chosen right-continuous. If $\lim\tau_O(\tau^n)=\tau_O(\tau)$ a.s.\ for any $\tau\in\T$ and sequence $(\tau^n)_{n\in\N}\subset\T$ with $\tau^n\searrow\tau$ a.s., then also $(L_t)$ can be chosen right-continuous.

All conditions are met when letting each $\tau_O(\tau)$ be the latest stopping time attaining the value of $F(\tau)$, or when letting each $\tau_O(\tau)=\tau$.
\end{lemma}

\begin{proof}
First rewrite $F(\tau)$ as
\begin{equation}\label{Fopt}
\begin{split}
F(\tau)=&\int_0^\tau\bigl(\pi^0_s-\pi^F_s\bigr)\,ds+E\biggl[\int_0^\infty\pi^B_s\,ds\biggv\F_\tau\biggr]+\esssup_{\tau'\geq \tau}E\biggl[\int_0^{\tau'}\bigl(\pi^F_s-\pi^B_s\bigr)\,ds\biggv\F_\tau\biggr].
\end{split}
\end{equation}
The first term on the RHS is a continuous process evaluated at $\tau$ which is by assumption adapted and bounded by $\int_0^\infty\bigl(\lvert\pi^0_s\rvert+\lvert\pi^F_s\rvert\bigr)\,ds \in L^1(P)$, hence optional and of class {\rm (D)}. The second and third terms are (super-)martingale-systems \citep[cf.][Proposition 2.26]{ElKaroui81} of class {\rm (D)}~-- particularly the latter bounded by the family $\bigl\{E\bigl[\int_0^\infty\bigl(\lvert\pi^F_s\rvert+\lvert\pi^B_s\rvert\bigr)\,ds\bigv\F_\tau\bigr],\tau\in\T\bigr\}$ of class {\rm (D)}. Thus there exist optional processes of class {\rm (D)} that aggregate the two \mbox{(super-)}martingale-systems, respectively. The former, being a martingale, may be chosen right-continuous. The latter is in fact the Snell envelope $U_Y$ of the continuous process $(Y_t):=(\int_0^t(\pi^F_s-\pi^B_s)\,ds)$, whence $U_Y$ is (right-)continuous in expectation and thus may be taken to have right-continuous paths, a.s.

$L(\tau)$ can be written like \eqref{Fopt}, with a third term $X(\tau):=E\bigl[\int_0^{\tau_O(\tau)}\bigl(\pi^L_s-\pi^B_s\bigr)\,ds\bigv\F_\tau\bigr]$. Suppose first that $\pi^L_s-\pi^B_s\geq 0$ for all $s\in\R_+$, a.s. In this case
\begin{align*}
E\bigl[X(\tau')\bigv\F_\tau\bigr]&=X(\tau)+E\biggl[\int_{\tau_O(\tau)}^{\tau_O(\tau')}\bigl(\pi^L_s-\pi^B_s\bigr)\,ds\biggv\F_\tau\biggr]\geq X(\tau)
\end{align*}
for all stopping times $\tau'\geq\tau$ (as $\tau_O(\tau')\geq\tau_O(\tau)$), so $X:=\bigl\{X(\tau),\tau\in\T\bigr\}$ is a submartingale-system. $X$ is bounded by $\bigl\{E\bigl[\int_0^\infty\bigl(\lvert\pi^L_s\rvert+\lvert\pi^B_s\rvert\bigr)\,ds\bigv\F_\tau\bigr],\tau\in\T\bigr\}$, hence of class {\rm (D)}. In general, the last argument applies separately to $\bigl(\pi^L_s-\pi^B_s\bigr)^+$ and $\bigl(\pi^L_s-\pi^B_s\bigr)^-$, showing that $X$ is the difference of two submartingale-systems which can be aggregated by two optional processes of class {\rm (D)}.

If $\lim\tau_O(\tau^n)=\tau_O(\tau)$ a.s.\ for any sequence $(\tau^n)_{n\in\N}\subset\T$ with $\tau^n\searrow\tau$ a.s., then $X$~-- being of class {\rm (D)}~-- is right-continuous in expectation and the aggregating submartingales can be chosen with right-continuous paths.

Finally, as the process $(Y_t)$ defined above is continuous, the latest stopping time after $\tau$ that attains $F(\tau)$~-- $\tau_F(\tau)$~-- is the first time the monotone part of the Snell envelope $U_Y$ increases. That monotone part inherits continuity from $(Y_t)$. Thus chosen, $\tau\leq\tau_F(\tau)\leq\tau_F(\tau')$ on $\{\tau\leq\tau'\}$ for all $\tau,\tau'\in\T$. Now consider a sequence of stopping times $\tau^n\searrow\tau$ a.s., whence also $\tau_F(\tau^n)$ decreases in $n$. By construction we can only have $\lim\tau_F(\tau^n)>\tau_F(\tau)\geq\tau$ where the monotone part of $U_Y$ is constant on $(\tau_F(\tau),\lim\tau_F(\tau^n)]$. By continuity it must then be constant on $[\tau_F(\tau),\lim\tau_F(\tau^n)]$. However, the monotone part of $U_Y$ increases at $\tau_F(\tau)$ by definition, so we must have $\tau_F(\tau)=\lim\tau_F(\tau^n)$ a.s.
\end{proof}

\begin{remark}\label{rem:LFladlag}
As the proof of Lemma \ref{lem:LFreg} relies on the aggregation of supermartingales of class {\rm (D)}, we may further assume that the processes $(L_t)$ and $(F_t)$ have left limits at any time $t$ \citep[see][Proposition 2.27]{ElKaroui81}.
\end{remark}

\begin{remark}\label{rem:optleader}
The solutions~-- and in particular the stopping regions~-- for the monopoly problem \eqref{Mono^i} and the problem \eqref{maxL} of when to become optimally the leader typically differ. Consider a model in which the profit streams are driven by a diffusion $(Y_t)$ such that each firm $i$ has a follower threshold, say $y^i_F$ solving \eqref{F} with $\tau^i_F(\tau)=\inf\{t\geq\tau\mid Y_t\geq y^i_F\}$, and firm $1$ also has a monopoly threshold, say $y^1_L\leq y^1_F$ solving \eqref{Mono^i}, and where $L^1_t$ can be represented as a continuous function of the state $Y_t$. Now one can apply arguments of \cite{Jacka93} relying on the semi-martingale property of $(L^1_t)$, which the proof of Lemma \ref{lem:LFreg} actually establishes. Denote the finite-variation part of $(L^1_t)$ by $(A_t)$. The Snell envelope $(S_t)$ of $(L^1_t)$, i.e.\ the value process of optimally stopping $(L^1_t)$, now is continuous (as a function of the state) as well and its monotone decreasing part $(B_t)$ is given by $dB_t=\indi{S_t=L^1_t}dA_t+\frac12dL^0_t(S_t-L^1_t)$. The last term is the local time of $(S_t-L^1_t)$ spent at $0$ (i.e.\ in the stopping region), which is absolutely continuous w.r.t.\ $\indi{S_t=L^1_t}dA_t\leq 0$.

Now suppose the stopping region $\{S_\cdot=L^1_\cdot\}$ is that of the monopoly problem, $\{Y_\cdot\geq y^1_L\}$, whence $dL^0_t(S_t-L^1_t)$ lives on the boundary $\{Y_\cdot=y^1_L\}$. For $Y_t\in[y^1_L,y^2_F)$, $(L^1_t)$ has a drift given by the foregone monopoly profit stream, $dA_t=-\pi^{L1}_t\,dt$, whence $dL^0_t(S_t-L^1_t)\equiv 0$ if $(Y_t)$ has a transition density, cf.\ Theorem 6 of \cite{Jacka93}.

As $(L^1_t)$ is of class {\rm (D)}, so is $(S_t)$, which thus converges to $S_\infty=L^1_\infty=0$ in $L^1(P)$ as $t\to\infty$. Therefore the martingale part of $(S_t)$ is simply $E[-B_\infty\mid\F_t]$ and $S_t=E[-\int_t^\infty\indi{S_s=L^1_s}\,dA_s\mid\F_t]$. Noting further that for $Y_t>y^2_F$, $(L^1_t)$ has a drift given by the foregone duopoly stream, $dA_t=-\pi^{B1}_t\,dt$, we then obtain
\begin{align}\label{Sleader}
S_t=E\biggl[\int_t^\infty\Bigl(\indi{Y_s\in[y^1_L,y^2_F)}\pi^{L1}_s+\indi{Y_s>y^2_F}\pi^{B1}_s\Bigr)\,ds-\int_t^\infty\indi{Y_s=y^2_F}\,dA_s\biggv\F_t\biggr].
\end{align}
By applying similar reasoning to firm $1$'s monopoly problem \eqref{Mono^i}, which is solved by $\tau^1_L(t)=\inf\{s\geq t\mid Y_s\geq y^1_L\}$, its value is
$E\bigl[\int_{\tau^1_L(t)}^\infty\pi^{L1}_s\,ds\bigv\F_t\bigr]=E\bigl[\int_t^\infty\indi{Y_s\geq y^1_L}\pi^{L1}_s\,ds\bigv\F_t\bigr]$, i.e.\ $E\bigl[\int_{\tau^1_L(t)}^\infty\indi{Y_s<y^1_L}\pi^{L1}_s\,ds\bigv\F_t\bigr]=0$. Thus, if $Y_t\geq y^1_L$, \eqref{Sleader} can be rewritten as
\begin{align*}
S_t=E\biggl[\int_t^\infty\Bigl(\indi{Y_s<y^2_F}\pi^{L1}_s+\indi{Y_s>y^2_F}\pi^{B1}_s\Bigr)\,ds-\int_t^\infty\indi{Y_s=y^2_F}\,dA_s\biggv\F_t\biggr].
\end{align*}
In this hypothesized stopping region for $(L^1_t)$, also $S_t=L^1_t$, in particular for $Y_t\geq y^2_F\geq y^1_L$,
\begin{align*}
S_t=E\biggl[\int_t^\infty\pi^{B1}_s\,ds\biggv\F_t\biggr].
\end{align*}
With $y^2_F$ in the stopping region, $-\indi{Y_s=y^2_F}\,dA_s\geq 0$, and by assumption $\pi^{L1}_\cdot\geq\pi^{B1}_\cdot$. Further, $\indi{Y_s=y^2_F}$ is a $P\otimes dt$ nullset if $Y$ has a transition density, such that equating the two last expressions for $S_t$ implies indeed
\[
E\biggl[\int_t^\infty\indi{Y_s<y^2_F}\Bigl(\pi^{L1}_s-\pi^{B1}_s\Bigr)\,ds\biggv\F_t\biggr]=0
\]
(and $E\bigl[-\int_t^\infty\indi{Y_s=y^2_F}\,dA_s\bigv\F_t\bigr]=0$). This contradicts the typical strict ordering $\pi^{L1}_\cdot>\pi^{B1}_\cdot$.
\end{remark}

\begin{lemma}\label{lem:thresconstr}
Let $(x_t)$ be a geometric Brownian motion on $\bigl(\Omega,\F,P\bigr)$, satisfying
\begin{equation*}
dx_t=\mu x_t\,dt+\sigma x_t\,dB_t
\end{equation*}
for a Brownian motion $(B_t)$ adapted to $\mathbb{F}$. Further let $\tau_{\tilde x}:=\inf\{t\geq 0\mid x_t\geq \tilde x\}$ for any given constant $\tilde x\in\R_+$. Then the problem
\begin{equation}\label{constrGBM}
\sup_{\tau\in\T,\,\tau\leq\tau_{\tilde x}}E\biggl[\int_\tau^\infty e^{-rt}(Dx_t-rI)\,dt\biggr]
\end{equation}
with $r>\max\{\mu,0\}$, $D\in\R$ and $I>0$ is solved by $\tau^*:=\inf\{t\geq 0\mid x_t\geq\tilde x\wedge x^*\}$, where
\[
x^*=\frac{\beta_1}{\beta_1-1}\cdot\frac{I(r-\mu)}{D^+}
\]
and $\beta_1>1$ is the positive root of $\frac12\sigma^2\beta(\beta-1)+\mu\beta-r=0$.
\end{lemma}

\begin{proof}
If $D\leq 0$, then the integrand in \eqref{constrGBM} is always negative and the latest feasible stopping time is optimal, which indeed satisfies $\tau_{\tilde x}=\tau^*$ as now $x^*=\infty$. For $D>0$, Lemma \ref{lem:thresconstr} is a special case of Proposition 4.6 in \cite{StegThijssen15}, setting their $Y_0=Dx_0$, $\mu_Y=\mu$, $\sigma_Y=\sigma$, $X_0=c_0=c_B=0$ and $y_\cP=(r-\mu_Y)(I-c_A/r)=\tilde x$.
\end{proof}

\section{Proofs}\label{app:proof}

\begin{proof}[{\bf Proof of Lemma \ref{lem:tau^1_F<tau^2_F}}]
The stopping problem in \eqref{F} is~-- up to a constant~-- equivalent to $\essinf_{\tau'\geq\tau}E\bigl[\int_\tau^{\tau'}(\pi^{Bi}_s-\pi^{Fi}_s)\,ds\bigv\F_\tau\bigr]$. Optimality of $\tau_F^i(\tau)$ and iterated expectations thus imply $E\bigl[\int_{\tau'}^{\tau_F^i(\tau)}(\pi^{Bi}_s-\pi^{Fi}_s)\,ds\bigv\F_{\tau'}\bigr]\leq 0$ for all $\tau'\in[\tau,\tau_F^i(\tau)]$ and $E\bigl[\int_{\tau_F^i(\tau)}^{\tau'}(\pi^{Bi}_s-\pi^{Fi}_s)\,ds\bigv\F_{\tau_F^i(\tau)}\bigr]\geq 0$ for all $\tau'\geq\tau_F^i(\tau)$, strictly on $\{\tau'>\tau_F^i(\tau)\}$ as $\tau^i_F(\tau)$ is the latest time attaining \eqref{F}. Thus, with $\tau'=\min\{\tau_F^1(\tau),\tau_F^2(\tau)\}$ and $\pi^{B2}_\cdot-\pi^{F2}_\cdot\leq\pi^{B1}_\cdot-\pi^{F1}_\cdot$ we have 
\begin{align*}
0\leq E\biggl[\int_{\tau'}^{\tau_F^1(\tau)}(\pi^{B2}_s-\pi^{F2}_s)\,ds\biggv\F_{\tau'}\biggr]\leq E\biggl[\int_{\tau'}^{\tau_F^1(\tau)}(\pi^{B1}_s-\pi^{F1}_s)\,ds\biggv\F_{\tau'}\biggr]\leq 0.
\end{align*}
The first inequality is strict on $\{\tau_F^2(\tau)<\tau_F^1(\tau)\}$ (up to a $P$-nullset), so $\tau^1_F(\tau)\leq\tau^2_F(\tau)$ ($P$-a.s.).

Finally, $F^i_{\tau}-M^i_{\tau}=\esssup_{\tau'\geq\tau}E[\int_\tau^{\tau'}(\pi^{Fi}_s-\pi^{Bi}_s)\,ds\mid\F_\tau]$ is not greater for $i=1$ than for $i=2$.
\end{proof}

\begin{proof}[{\bf Proof of Lemma \ref{lem:L-F}}]
We have
\begin{flalign}
&& L^2_\tau-F^2_\tau&=E\biggl[\int_{\tau}^{\tau^1_F(\tau)}(\pi^{L2}_s-\pi^{F2}_s)\,ds+\int_{\tau^1_F(\tau)}^{\tau^2_F(\tau)}(\pi^{B2}_s-\pi^{F2}_s)\,ds\biggv\F_\tau\biggr] && \label{L^2-F^2}\\
&\text{and} \nonumber\\
&& L^1_\tau-F^1_\tau&=E\biggl[\int_{\tau}^{\tau^1_F(\tau)}(\pi^{L1}_s-\pi^{F1}_s)\,ds+\int_{\tau^1_F(\tau)}^{\tau^2_F(\tau)}(\pi^{L1}_s-\pi^{B1}_s)\,ds\biggv\F_\tau\biggr], && \nonumber
\end{flalign}
where $\tau^1_F(\tau)\leq\tau^2_F(\tau)$ by Lemma \ref{lem:tau^1_F<tau^2_F}. By the optimality of $\tau^2_F(\tau)$ for stopping the stream $(\pi^{B2}_s-\pi^{F2}_s)$, the second integral on the RHS of \eqref{L^2-F^2} has non-positive conditional expectation, cf.\ the proof of Lemma \ref{lem:tau^1_F<tau^2_F}. The claim now follows from the assumptions $\pi^{L1}_\cdot-\pi^{F1}_\cdot\geq\pi^{L2}_\cdot-\pi^{F2}_\cdot$ and $\pi^{L1}_\cdot\geq\pi^{B1}_\cdot$.
\end{proof}

\begin{proof}[{\bf Proof of Lemma \ref{lem:L<Fwait}}]
We only use the assumptions $\pi^{Li}_\cdot\geq\pi^{Bi}_\cdot$ and $\pi^{0i}_\cdot\geq\pi^{Fi}_\cdot$ (except for the representation with $\tau_\cP(\vartheta)$). Let $\tau^i_{1\text{st}}(\vartheta)=\inf\{t\geq\vartheta\mid L^i_t>F^i_t\}$ ($=\tau_\cP(\vartheta)$ for $i=2$), such that $M^i_\cdot\leq L^i_\cdot\leq F^i_\cdot$ on $[\vartheta,\tau^i_{1\text{st}}(\vartheta))$, so investing is nowhere better than becoming follower, but indeed inferior if the last inequality is strict. Next, by the optimality of $\tau^i_F(\vartheta)$ in $F^i_\vartheta$ and $\pi^{0i}_\cdot\geq\pi^{Fi}_\cdot$, $F^i_\cdot$ is nondecreasing in expectation on $[\vartheta,\tau^i_F(\vartheta)]$, so firm $i$ prefers to become follower as late as possible on that interval. Finally, $L^i_\tau\geq F^i_\tau$ at $\tau=\min\bigl\{\tau^i_{1\text{st}}(\vartheta),\tau^i_F(\vartheta)\bigr\}$~-- at $\tau^i_{1\text{st}}(\vartheta)$ due to right-continuity of $L^i_\cdot-F^i_\cdot$ and at $\tau^i_F(\vartheta)$ due to $\pi^{Li}_\cdot\geq\pi^{Bi}_\cdot$. Thus, in case the opponent does not invest before $\tau=\min\bigl\{\tau^i_{1\text{st}}(\vartheta),\tau^i_F(\vartheta)\bigr\}$ (with some probability), firm $i$ can reach at least its follower value there by the limit from planning to invest at $\tau+1/n$ and $n\to\infty$ (in the limit, firm $i$ obtains $F^i_\tau$ with the probability that the opponent invests at $\tau$ and $L^i_\tau$ else as $L^i_\cdot$ is right-continuous).
\end{proof}

\begin{proof}[{\bf Proof of Lemma \ref{lem:Lmaxstop}}]
Where $L^i_\vartheta>E[L^i_\tau\mid\F_\vartheta]$ for all stopping times $\tau>\vartheta$, there we must also have $L^i_\vartheta\geq E\bigl[F^i_\tau\bigv\F_\vartheta\bigr]$ for any $\tau\geq\vartheta$, strictly on $\{\tau>\vartheta\}$, as follows. First note that $F^i_\tau-E\bigl[F^i_{\tau_F^i(\tau)}\bigv\F_\tau\bigr]=E\bigl[\int_\tau^{\tau_F^i(\tau)}(\pi^{Fi}_s-\pi^{0i}_s)\bigv\F_\tau\bigr]\leq 0$ because $\tau_F^i(\tau_F^i(\tau))=\tau_F^i(\tau)$. Furthermore note that $L^i_{\tau_F^i(\tau)}\geq F^i_{\tau_F^i(\tau)}$ by $\pi^{Li}_\cdot\geq\pi^{Bi}_\cdot$. Together with the hypothesis it must thus hold that $L^i_\vartheta>E\bigl[F^i_\tau\bigv\F_\vartheta\bigr]\geq E\bigl[M^i_\tau\bigv\F_\vartheta\bigr]$ on $\{\tau>\vartheta\}$ for any $\tau\in\T$, and  $L^i_\vartheta\geq F^i_\vartheta\geq M^i_\vartheta$ using $\tau=\vartheta$. 

Then, in case that the opponent's plan does not imply immediate investment with probability $1$ (else there is nothing to prove), firm $i$ cannot achieve a higher payoff than $L^i_\vartheta$ with the probability that firm $j$ does not invest immediately and $F^i_\vartheta$ with the probability that firm $j$ invests immediately. 
Thanks to right-continuity of $L^i_\cdot$, that upper bound is the limit of firm $i$ planning to invest at $\vartheta+1/n$ and $n\to\infty$, but it is not attainable by any plan that does not induce immediate investment with probability $1$.

For the second claim suppose by way of contradiction that $\tau=\vartheta$ attains \eqref{maxL}, but that there exists a stopping time $\tau'\geq\vartheta$ such that $E\big[\int_{\vartheta}^{\tau'}(\pi^{Li}_s-\pi^{0i}_s)\,ds\bigv\F_{\vartheta}\bigr]<0$ with positive probability. On that event,
\begin{align*}
L^i_\vartheta&=\int_0^\vartheta\pi^{0i}_s\,ds+E\biggl[\int_\vartheta^{\tau^j_F(\vartheta)}\pi^{Li}_s\,ds+\int_{\tau^j_F(\vartheta)}^\infty\pi^{Bi}_s\,ds\biggv\F_\vartheta\biggr]\\
&<\int_0^\vartheta\pi^{0i}_s\,ds+E\biggl[\int_\vartheta^{\tau'}\pi^{0i}_s\,ds+\int_{\tau'}^{\tau^j_F(\vartheta)}\pi^{Li}_s\,ds+\int_{\tau^j_F(\vartheta)}^\infty\pi^{Bi}_s\,ds\biggv\F_\vartheta\biggr]\leq E\Bigl[L^i_{\tau'}\Bigv\F_\vartheta\Bigr]
\end{align*}
as $\tau^j_F(\tau')\geq\tau^j_F(\vartheta)$ and $\pi^{Li}_\cdot\geq\pi^{Bi}_\cdot$, which contradicts the optimality of $\tau=\vartheta$ in \eqref{maxL}.
\end{proof}

\begin{remark}
The $\F$-events on which $\tau>\vartheta\Rightarrow L^i_\vartheta>E[L^i_\tau\mid\F_\vartheta]$ a.s.\ for all stopping times $\tau\geq\vartheta$ can be aggregated into an $\F_\vartheta$-event as follows: With $A(\tau):=\{\tau>\vartheta\}\in\F_\vartheta$ and $B(\tau):=\{L^i_\vartheta>E[L^i_\tau\mid\F_\vartheta]\}\in\F_\vartheta$ for any stopping time $\tau\geq\vartheta$, the given property can be written as $\indi{B(\tau)}-\indi{A(\tau)}=0$ a.s.\ for all $\tau\geq\vartheta$ (as $B(\tau)\subset A(\tau)$). The latter holds for any $\F$-event if and only if it is a subset of $C_0:=\{\essinf_{\tau\geq\vartheta}(\indi{B(\tau)}-\indi{A(\tau)})=0\}$ (up to a nullset). As all $\indi{B(\tau)}-\indi{A(\tau)}$ are $\F_\vartheta$-measurable random variables, so is $\essinf_{\tau\geq\vartheta}(\indi{B(\tau)}-\indi{A(\tau)})$. Indeed, as $\indi{B(\tau)}-\indi{A(\tau)}\geq\essinf_{\tau\geq\vartheta}(\cdot)$, also $\indi{B(\tau)}-\indi{A(\tau)}\geq E[\essinf_{\tau\geq\vartheta}(\cdot)\mid\F_\vartheta]$ a.s.\ for all $\tau\geq\vartheta$ and thus $\essinf_{\tau\geq\vartheta}(\cdot)\geq E[\essinf_{\tau\geq\vartheta}(\cdot)\mid\F_\vartheta]$ a.s.\ by the definition of $\essinf(\cdot)$. However, as the left and right-hand sides have the same expectation, equality holds a.s.

Further, there exists a sequence of mutually disjoint sets $(C_n)$ and a sequence of stopping times $(\tau_n)$ such that $\bigcup C_n=\Omega\setminus C_0$ (up to a nullset), $\inf\tau_n\geq\vartheta$ and, on each $C_n$, $\tau_n>\vartheta$ and $L^i_\vartheta=E[L^i_{\tau_n}\mid\F_\vartheta]$ a.s. This follows from the fact that the family $\{\indi{B(\tau)}-\indi{A(\tau)}\mid\tau\geq\vartheta\}$ is directed downwards, as by all $\indi{B(\tau)}-\indi{A(\tau)}$ being $\{-1,0\}$-valued, for any $\tau_1,\tau_2\geq\vartheta$ also $\tau_3:=\tau_1+(\indi{A(\tau_2)}-\indi{B(\tau_2)})(\tau_2-\tau_1)\geq\vartheta$ is a stopping time that satisfies $\indi{A(\tau_3)}-\indi{B(\tau_3)}=\min(\indi{A(\tau_1)}-\indi{B(\tau_1)},\indi{A(\tau_2)}-\indi{B(\tau_2)})$. Thus there exists a sequence $(\tau_n)\subset\T$ with $\inf\tau_n\geq\vartheta$ and $\indi{B(\tau_n)}-\indi{A(\tau_n)}\searrow\essinf_{\tau\geq\vartheta}(\indi{B(\tau)}-\indi{A(\tau)})$ a.s., so $P[\{\indi{B(\tau_n)}=\indi{A(\tau_n)}\}\setminus C_0]\searrow 0$. Now one can recursively set $C_n=A(\tau_n)\setminus(B(\tau_n)\cup C_{n-1})$.
\end{remark}

\begin{proof}[{\bf Proof of Lemma \ref{lem:stoptau^1_F}}]
First note that there exists an optimal stopping time for \eqref{mono1} (and also a latest one), because the process to be stopped is continuous and integrable. For any stopping time $\tau\in[\vartheta,\tau^2_F(\vartheta)]$, $\tau^2_F(\tau)=\tau^2_F(\vartheta)$ and thus $L^1_\vartheta-E\bigl[L^1_\tau\bigv\F_\vartheta\bigr]=E\bigl[\int_\vartheta^\tau(\pi^{L1}_s-\pi^{01}_s)\,ds\bigv\F_\vartheta\bigr]$ is the same payoff difference as that between $\vartheta$ and $\tau$ in \eqref{mono1}. Thus, where $\vartheta$ is uniquely optimal in \eqref{mono1}, there also $L^1_\vartheta>E\bigl[L^1_\tau\bigv\F_\vartheta\bigr]$ on $\{\tau>\vartheta\}$. Regarding the other possible payoffs, as argued in the proof of Lemma \ref{lem:Lmaxstop}, $M^1_\tau\leq F^1_\tau\leq E\bigl[F^1_{\tau_F^1(\tau)}\bigv\F_\tau\bigr]\leq E\bigl[L^1_{\tau_F^1(\tau)}\bigv\F_\tau\bigr]$, where now $\tau_F^1(\tau)\leq\tau_F^2(\tau)=\tau_F^2(\vartheta)$ for $\tau\in[\vartheta,\tau^2_F(\vartheta)]$. Hence $L^1_\vartheta$ is strictly superior to any future payoff on $(\vartheta,\tau^2_F(\vartheta)]$ and the game has to end by the same arguments as in the proof of Lemma \ref{lem:Lmaxstop}.
\end{proof}

\begin{proof}[{\bf Proof of Lemma \ref{lem:Fmaxstop}}]
First note that there exists an optimal stopping time $\tau_M^i\geq\vartheta$ for \eqref{sim_i} and also a latest one, because the process to be stopped is continuous and integrable. An optimal $\tau_M^i$ satisfies the necessary and sufficient conditions $E\bigl[\int_\tau^{\tau_M^i}(\pi^{0i}_s-\pi^{Bi}_s)\,ds\bigv\F_\tau\bigr]\geq 0$ on $\{\tau\leq\tau_M^i\}$ and $E\bigl[\int_{\tau_M^i}^\tau(\pi^{0i}_s-\pi^{Bi}_s)\,ds\bigv\F_{\tau_M^i}\bigr]\leq 0$ on $\{\tau\geq\tau_M^i\}$ for all stopping times $\tau\geq\vartheta$, the last inequality being strict on $\{\tau>\tau_M^i\}$ if $\tau_M^i$ is the latest solution. We will derive the analogous properties for the process $(F^i_t)$; thus consider an arbitrary stopping time $\tau\geq\vartheta$.

For the first property, note that on $\{\tau\leq\tau_M^i\}$ we have 
\[
E\bigl[F^i_{\tau_M^i\wedge\tau_F^i(\tau)}\bigv\F_\tau\bigr]-F^i_\tau=E\biggl[\int_\tau^{\tau_M^i\wedge\tau_F^i(\tau)}(\pi^{0i}_s-\pi^{Fi}_s)\,ds\biggv\F_\tau\biggr]\geq 0
\] 
by $\pi^{0i}_\cdot\geq\pi^{Fi}_\cdot$ and $\tau^i_F(\tau_M^i\wedge\tau_F^i(\tau))=\tau^i_F(\tau)$. Further, on the subset $\{\tau_M^i>\tau_F^i(\tau)\}$ we have
\[
E\bigl[F^i_{\tau_M^i}\bigv\F_{\tau_F^i(\tau)}\bigr]-F^i_{\tau_F^i(\tau)}=E\biggl[\int_{\tau_F^i(\tau)}^{\tau_M^i}(\pi^{0i}_s-\pi^{Bi}_s)\,ds+\int_{\tau_M^i}^{\tau_F^i(\tau_M^i)}(\pi^{Fi}_s-\pi^{Bi}_s)\,ds\biggv\F_{\tau_F^i(\tau)}\biggr]\geq 0
\]
by the optimality of $\tau_M^i$ and the definition of $\tau^i_F(\tau_M^i)$, cf.\ the proof of Lemma \ref{lem:tau^1_F<tau^2_F}. Together, $E\bigl[F^i_{\tau_M^i}\bigv\F_\tau\bigr]-F^i_\tau=E\bigl[F^i_{\tau_M^i}-F^i_{\tau_M^i\wedge\tau_F^i(\tau)}\bigv\F_\tau\bigr]+E\bigl[F^i_{\tau_M^i\wedge\tau_F^i(\tau)}\bigv\F_\tau\bigr]-F^i_\tau\geq 0$.

For the second property, note that $E\bigl[F^i_{\tau_F^i(\tau)}\bigv\F_\tau\bigr]-F^i_\tau=E\bigl[\int_\tau^{\tau_F^i(\tau)}(\pi^{0i}_s-\pi^{Fi}_s)\,ds\bigv\F_\tau\bigr]\geq 0$, again by $\pi^{0i}_\cdot\geq\pi^{Fi}_\cdot$ and $\tau^i_F(\tau_F^i(\tau))=\tau^i_F(\tau)$, hence it is sufficient to show $E\bigl[F^i_{\tau_F^i(\tau)}\bigv\F_{\tau_M^i}\bigr]\leq F^i_{\tau_M^i}$ on $\{\tau\geq\tau_M^i\}$. There, where $\tau_F^i(\tau)\geq\tau_F^i(\tau_M^i)$, it holds that
\begin{align*}
E\bigl[F^i_{\tau_F^i(\tau)}\bigv\F_{\tau_M^i}\bigr]-F^i_{\tau_M^i}&=E\biggl[\int_{\tau_M^i}^{\tau_F^i(\tau_M^i)}(\pi^{0i}_s-\pi^{Fi}_s)\,ds+\int_{\tau_F^i(\tau_M^i)}^{\tau_F^i(\tau)}(\pi^{0i}_s-\pi^{Bi}_s)\,ds\biggv\F_{\tau_M^i}\biggr]\\
&\leq E\biggl[\int_{\tau_M^i}^{\tau_F^i(\tau_M^i)}(\pi^{0i}_s-\pi^{Bi}_s)\,ds+\int_{\tau_F^i(\tau_M^i)}^{\tau_F^i(\tau)}(\pi^{0i}_s-\pi^{Bi}_s)\,ds\biggv\F_{\tau_M^i}\biggr]\leq 0,
\end{align*}
where we have used the definition of $\tau^i_F(\tau_M^i)$ in the first estimate, and the optimality of $\tau_M^i$ in the last. The last inequality is strict on $\{\tau>\tau_M^i\}$ if $\tau_M^i$ is the latest solution of \eqref{sim_i}.

Now suppose that the stopping time $\tau_M^i\geq\vartheta$ optimally stops $(F^i_t)$ from $\vartheta\in\T$, i.e.\ it satisfies $E\bigl[F^i_{\tau_M^i}\bigv\F_\tau\bigr]\geq F^i_\tau$ on $\{\tau\leq\tau_M^i\}$ and $E\bigl[F^i_\tau\bigv\F_{\tau_M^i}\bigr]\leq F^i_{\tau_M^i}$ on $\{\tau\geq\tau_M^i\}$ for all stopping times $\tau\geq\vartheta$. As $E\bigl[F^i_{\tau_F^i(\tau_M^i)}\bigv\F_{\tau_M^i}\bigr]\geq F^i_{\tau_M^i}$ as noted above, we must then have equality, i.e.\ $\tau_F^i(\tau_M^i)$ is optimal, too, and we may set $\tau_M^i=\tau_F^i(\tau_M^i)$ for simplicity to show optimality of $\tau_F^i(\tau_M^i)$ in \eqref{sim_i}. Therefore, consider again an arbitrary stopping time $\tau\geq\vartheta$.

On $\{\tau\leq\tau_M^i\}$, where $\tau_F^i(\tau)\leq\tau_F^i(\tau_M^i)=\tau_M^i$, it then holds that
\begin{align*}
0\leq E\bigl[F^i_{\tau_M^i}\bigv\F_{\tau}\bigr]-F^i_{\tau}&=E\biggl[\int_{\tau}^{\tau_F^i(\tau)}(\pi^{0i}_s-\pi^{Fi}_s)\,ds+\int_{\tau_F^i(\tau)}^{\tau_M^i}(\pi^{0i}_s-\pi^{Bi}_s)\,ds\biggv\F_{\tau}\biggr]\\
&\leq E\biggl[\int_{\tau}^{\tau_F^i(\tau)}(\pi^{0i}_s-\pi^{Bi}_s)\,ds+\int_{\tau_F^i(\tau)}^{\tau_M^i}(\pi^{0i}_s-\pi^{Bi}_s)\,ds\biggv\F_{\tau}\biggr]
\end{align*}
by the definition of $\tau_F^i(\tau)$, which yields the first optimality property for $\tau_M^i$ in \eqref{sim_i}.

On $\{\tau\geq\tau_M^i\}$, where $\tau_F^i(\tau)\geq\tau_M^i$, we have
\begin{align*}
0\geq E\bigl[F^i_{\tau}\bigv\F_{\tau_M^i}\bigr]-F^i_{\tau_M^i}&=E\biggl[\int_{\tau_M^i}^{\tau}(\pi^{0i}_s-\pi^{Bi}_s)\,ds+\int_{\tau}^{\tau_F^i(\tau)}(\pi^{Fi}_s-\pi^{Bi}_s)\,ds\biggv\F_{\tau_M^i}\biggr]\\
&\geq E\biggl[\int_{\tau_M^i}^{\tau}(\pi^{0i}_s-\pi^{Bi}_s)\,ds\biggv\F_{\tau_M^i}\biggr]
\end{align*}
again by the definition of $\tau_F^i(\tau)$, yielding the second optimality property for $\tau_M^i$ in \eqref{sim_i}.
\end{proof}

\begin{proof}[{\bf Proof of Proposition \ref{prop:PKpreemp}}]
By the strong Markov property it suffices to consider $t=0$. If the preemption region is empty, one can set $\ubar x=\bar x$ and pick any number in $(0,x_F^1]$. The upper and lower bounds for a non-empty preemption region are obtained as follows. First note that $L^2_0=M^2_0\leq F^2_0$ for all $x_0\geq x_F^1$. Second, for all $x_0>0$, $L^2_0\leq E\bigl[\int_0^\infty e^{-rs}\bigl(x_sD_{10}-rI^2\bigr)\,ds\bigr]=x_0D_{10}/({r-\mu})-I^2$ by $D_{10}\geq D_{11}$ and $F^2_0\geq E\bigl[\int_0^\infty e^{-rs}x_sD_{01}\,ds\bigr]=x_0D_{01}/({r-\mu})$, the value of never investing as follower. Thus, $L^2_0-F^2_0\leq x_0(D_{10}-D_{01})/({r-\mu})-I^2\leq 0$ on the non-empty interval $(0,({r-\mu})I^2/(D_{10}-D_{01})^+)$.

Now suppose $L^2_0>F^2_0$ for some $x_0=\hat x\in(0,x_F^1)$ and also for some $x_0=\check x<\hat x$, and assume by way of contradiction that $L^2_0\leq F^2_0$ for $x_0=x'\in(\check x,\hat x)$. Then we must have $x'>rI^2/(D_{10}-D_{01})^+$, because otherwise $L^2_0-F^2_0=E\bigl[\int_0^{\tau'}e^{-rs}\bigl(x_s(D_{10}-D_{01})-rI^2\bigr)\,ds\bigr]+E\bigl[L^2_{\tau'}-F^2_{\tau'}\bigr]\leq 0$ if $x_0=\check x$ and $x'\in(\check x,rI^2/(D_{10}-D_{01})^+\wedge x_F^1]$, where $\tau':=\inf\{s\geq 0\mid x_s\geq x'\}\leq\tau_F^1(0)$. By the same argument, we must also have $L^2_0>F^2_0$ for $x_0=\check x\vee rI^2/(D_{10}-D_{01})<x'$. But then, if we set $x_0=x'$ and $\hat\tau:=\inf\{s\geq 0\mid x_s\not\in(\check x\vee rI^2/(D_{10}-D_{01}),\hat x)\}\leq\tau_F^1(0)$, we obtain $L^2_0-F^2_0=E\bigl[\int_0^{\hat\tau}e^{-rs}\bigl(x_s(D_{10}-D_{01})-rI^2\bigr)\,ds\bigr]+E\bigl[L^2_{\hat\tau}-F^2_{\hat\tau}\bigr]>0$, whence the set $\{x>0\mid L^2_0>F^2_0\text{ given }x_0=x\}$ is convex. Further, that set is open as $L^2_0-F^2_0$ is continuous in $x_0$.

Suppose finally that $I^2=I^1$ and that the preemption region is non-empty, i.e., by Lemma \ref{lem:PreemMono} and the discussion thereafter, that the threshold solving \eqref{maxL2-F2} satisfies  $x_\Delta^2<x_F^1=x_F^2$. Then, for any $x_0\in[x_\Delta^2,x_F^2)$, $L^2_0-F^2_0=E\bigl[\int_0^{\tau_F^2(0)}\bigl(x_s(D_{10}-D_{01})-rI^2\bigr)\,ds\bigr]>0$ as $x_\Delta^2$ solves \eqref{maxL2-F2} uniquely.
\end{proof}

\begin{proof}[{\bf Proof of Proposition \ref{prop:lowconstr}}]
$\bar x<x_F^2$ can be any two numbers from $(0,\infty]$ in this proof, i.e., we only assume $\bar x$ finite. For initial states $x_0\in(\bar x,x_F^2)$, the constraint $\tau_\cP(0)\wedge\tau_F^2(0)$ in problem \eqref{MonoPreemPK} is the exit time from the given interval and \eqref{MonoPreemPK} is equivalent to
\begin{equation}\label{MarkovConstr}
\sup_{\tau\leq\inf\{s\geq 0\mid x_s\not\in(\bar x,x_F^2)\}}E\biggl[\int_\tau^\infty e^{-rs}\bigl(x_s(D_{10}-D_{00})-rI^1\bigr)\,ds\biggr].
\end{equation}
If $\bar x(D_{10}-D_{00})\geq rI^1$, the expected payoff difference between stopping at time $0$ and any feasible $\tau\geq 0$ is $E\bigl[\int_0^{\tau}e^{-rs}(x_s(D_{10}-D_{00})-rI^1)\,ds\bigr]\geq 0$, such that immediate stopping is optimal. If $D_{10}-D_{00}\leq 0$, also $E\bigl[\int_\tau^{\tau_\cP(0)\wedge\tau_F^2(0)}e^{-rs}(x_s(D_{10}-D_{00})-rI^1)\,ds\bigr]\leq 0$ for any $\tau\leq\tau_\cP(0)\wedge\tau_F^2(0)$, such that waiting until the constraint is optimal.


Now suppose $0<\bar x(D_{10}-D_{00})<rI^1$, whence $D_{10}>D_{00}$ and $x_L^1<\infty$. Note that
\[
E\biggl[\int_0^\infty e^{-rs}\bigl(x_s(D_{10}-D_{00})-rI^1\bigr)\,ds\biggr]=x_0\frac{D_{10}-D_{00}}{r-\mu}-I^1
\]
is the value of stopping immediately in \eqref{MarkovConstr}. Letting $x_0=x$, we will first verify that the value function of problem \eqref{MarkovConstr} is
\begin{equation}\label{valfct}
V(x):=\begin{cases} A(\hat x)x^{\beta_1}+B(\hat x)x^{\beta_2} & \text{ if }x\in(\bar x,\hat x), \\ x\frac{D_{10}-D_{00}}{r-\mu}-I^1 & \text{ else, } \end{cases}
\end{equation}
and thus $(\bar x,\hat x)^c$ the sought stopping region, under the hypothesis that either $\hat x\in[rI^1/(D_{10}-D_{00}),x_F^2)$ solves \eqref{hatx} or ``$\leq$'' holds for $\hat x=x_F^2$. Afterwards we will establish existence of a unique such $\hat x$. 

$V(x)$ as defined in \eqref{valfct} is continuous because $A(\hat x)$ and $B(\hat x)$ given by \eqref{AB} are the solution to the continuity conditions
\begin{align}\label{valmatch}
A\bar x^{\beta_1}+B\bar x^{\beta_2}&=\bar x\frac{D_{10}-D_{00}}{r-\mu}-I^1, \nonumber\\ 
A\hat x^{\beta_1}+B\hat x^{\beta_2}&=\hat x\frac{D_{10}-D_{00}}{r-\mu}-I^1.
\end{align}
$V(x)$ is also twice continuously differentiable on $(\bar x,x_F^2)$, except possibly at $\hat x$. At $\hat x<x_F^2$, the first derivative of $V$ is continuous, however, because \eqref{hatx} is the differentiability condition $\beta_1A\hat x^{\beta_1-1}+\beta_2B\hat x^{\beta_2-1}=(D_{10}-D_{00})/(r-\mu)$ multiplied by $\hat x$, minus the second continuity condition in \eqref{valmatch}. Therefore one can apply It\=o's lemma to see that $(e^{-rt}V(x_t))$ is a continuous, bounded supermartingale until $\tau=\inf\{t\geq 0\mid x_t\not\in(\bar x,x_F^2)\}$, with zero drift for $x_t\in(\bar x,\hat x)$ and drift $e^{-rt}(rI^1-x_t(D_{10}-D_{00}))\,dt<0$ for $x_t\in(\hat x,x_F^2)$. As that supermartingale coincides with the payoff process at $\tau=\inf\{t\geq 0\mid x_t\not\in(\bar x,x_F^2)\}$, it remains to show that $V(x)$ dominates the payoff process for $x\in(\bar x,x_F^2)$, which it does by construction for $x\in[\hat x,x_F^2]$.

For $x\in(\bar x,\hat x)$, $V''(x)=x^{\beta_2-2}\bigl[\beta_1(\beta_1-1)A(\hat x)x^{\beta_1-\beta_2}+\beta_2(\beta_2-1)B(\hat x)\bigr]$. As $\beta_k(\beta_k-1)>0$, $k=1,2$, the difference $V(x)-x(D_{10}-D_{00})/(r-\mu)+I^1$ would be convex if $A(\hat x),B(\hat x)\geq 0$, and it vanishes at both ends $\bar x,\hat x$. By \eqref{hatx}, that difference's derivative is non-positive at $\hat x$, where the difference would thus take its minimum. Hence it would vanish on all of $[\bar x,\hat x]$, but $V(x)$ cannot be affine on non-empty $(\bar x,\hat x)$. So we must have $A(\hat x)\wedge B(\hat x)<0$. If we had $B(\hat x)\geq 0$, then $A(\hat x)<0$ and $V(x)$ would be strictly decreasing on $(\bar x,\hat x)$, contradicting $V(\hat x)\geq V(\bar x)$; thus $B(\hat x)<0$. Going back to $V''(x)$, which can switch sign at most once, it must start strictly negative at $\bar x$. If it stays non-positive, the difference $V(x)-x(D_{10}-D_{00})/(r-\mu)+I^1$ is concave and thus non-negative on $(\bar x,\hat x)$. If $V''(x)$ eventually becomes positive, then the convex part of $V(x)-x(D_{10}-D_{00})/(r-\mu)+I^1$ takes its minimum $0$ at $\hat x$ as argued before, such that the difference is non-negative at the transition, and thus non-negative for the first, concave part. In summary, $(e^{-rt}V(x_t))$ is a supermartingale until $x_t$ leaves $(\bar x,x_F^2)$, dominating the payoff $e^{-rt}(x_t(D_{10}-D_{00})/(r-\mu)-I^1)$, which it coincides with for $x_t\in\{\bar x\}\cup[\hat x,x_F^2]$, so the latter is the stopping set in $[\bar x,x_F^2]$.

Next, we show that there is a unique threshold $\hat x\in[rI^1/(D_{10}-D_{00}),x_L^1)$ solving \eqref{hatx}, and then finally consider the constraint $x_F^2$. 

As the first step, note that $B(x)<0$ in \eqref{AB} for all $x\in(\bar x,x_L^1]$. Indeed, as the first term $\bigl[\bar x^{\beta_1}x^{\beta_2}-x^{\beta_1}\bar x^{\beta_2}\bigr]^{-1}<0$ for $x>\bar x$ by $\beta_1>1$ and $\beta_2<0$, we have $B(x)<0\Leftrightarrow x^{-\beta_1}\bigl[x(D_{10}-D_{00})/(r-\mu)-I^1\bigr]>\bar x^{-\beta_1}\bigl[\bar x(D_{10}-D_{00})/(r-\mu)-I^1\bigr]$. The derivative of the latter function of $x$ can be written as $x^{-\beta_1-1}\bigl[\beta_1I^1-(\beta_1-1)x(D_{10}-D_{00})/(r-\mu)\bigr]>0$ for all $x<x_L^1=\beta_1(r-\mu)I^1/((\beta_1-1)(D_{10}-D_{00}))$.

As the second step, note that with $A=A(x_L^1)$ and $B=B(x_L^1)$, we have $A\cdot(x_L^1)^{\beta_1}+B\cdot(x_L^1)^{\beta_2}=I^1/(\beta_1-1)$ by using the definition of $x_L^1$ in \eqref{valmatch}, and thus $(\beta_1-1)A\cdot(x_L^1)^{\beta_1}+(\beta_2-1)B\cdot(x_L^1)^{\beta_2}=I^1+(\beta_2-\beta_1)B\cdot(x_L^1)^{\beta_2}>I^1$ compared to ``$=$'' in \eqref{hatx}.

The third step is to show that ``$\leq$'' holds in \eqref{hatx} for the candidate $\hat x=rI^1/(D_{10}-D_{00})\in(\bar x,x_F^2)$, where the inclusion is exactly the current considered case. By similar arguments as above, using the continuity condition \eqref{valmatch}, $V(x)$ then satisfies
\[
V(x)=E\biggl[\int_{\hat\tau}^\infty e^{-rs}\bigl(x_s(D_{10}-D_{00})-rI^1\bigr)\,ds\biggr], \quad x_0=x\in[\bar x,\hat x],
\]
where we let $\hat\tau:=\inf\{s\geq 0\mid x_s\not\in(\bar x,\hat x)\}$. For $\hat x=rI^1/(D_{10}-D_{00})$, the integrand would be strictly negative until $\hat\tau$, so $V(x)>x(D_{10}-D_{00})/(r-\mu)-I^1$ for all $x\in(\bar x,\hat x)$. At $x=\hat x$, however, equality holds by \eqref{valmatch} and thus $V'(\hat x-)=\beta_1A(\hat x)\hat x^{\beta_1-1}+\beta_2B(\hat x)\hat x^{\beta_2-1}\leq (D_{10}-D_{00})/(r-\mu)$. Together with \eqref{valmatch}, the latter inequality implies also ``$\leq$'' in \eqref{hatx}. 

As the last step, as the function $(\beta_1-1)A(x)x^{\beta_1}+(\beta_2-1)B(x)x^{\beta_2}$ is continuous, it must attain $I^1$ at some $\hat x\in[rI^1/(D_{10}-D_{00}),x_L^1)$ by the second and third steps. The latter interval is non-empty by the estimate for $x_L^1$ at the beginning of the proof.

Concerning uniqueness, suppose $\hat x_1,\hat x_2\in[rI^1/(D_{10}-D_{00}),x_L^1)$ solve \eqref{hatx}. With either solution, as we have proved above, $V(x)$ is the value function of problem \eqref{MarkovConstr} for any $x_F^2\geq x_L^1$, and \eqref{MarkovConstr} is solved by both $\hat\tau_k:=\inf\{s\geq 0\mid x_s\not\in(\bar x,\hat x_k)\}$, $k=1,2$. In particular, for any $x_0\in[x_1,x_2]$,
\begin{align*}
&& V(x_0)=x_0\frac{D_{10}-D_{00}}{r-\mu}-I^1&=E\biggl[\int_{\hat\tau_2}^\infty e^{-rs}\bigl(x_s(D_{10}-D_{00})-rI^1\bigr)\,ds\biggr] && \displaybreak[0]\\
&\Rightarrow & 0&=E\biggl[\int_0^{\hat\tau_2}e^{-rs}\bigl(x_s(D_{10}-D_{00})-rI^1\bigr)\,ds\biggr]. &&
\end{align*}
Thus, letting $\check\tau_1:=\inf\{s\geq 0\mid x_s\leq\hat x_1\}\leq\hat\tau_2$ and still $x_0\in[x_1,x_2]$, 
\begin{align*}
0&=E\biggl[\int_0^{\hat\tau_2} e^{-rs}\bigl(x_s(D_{10}-D_{00})-rI^1\bigr)\,ds\biggr] \\
&=E\biggl[\int_0^{\check\tau_1\wedge\hat\tau_2}e^{-rs}\bigl(x_s(D_{10}-D_{00})-rI^1\bigr)\,ds+\int_{\check\tau_1\wedge\hat\tau_2}^{\hat\tau_2}e^{-rs}\bigl(x_s(D_{10}-D_{00})-rI^1\bigr)\,ds\biggr].
\end{align*}
The second integral vanishes itself in expectation, whereas the first integrand is strictly positive for $x_s\in(\hat x_1,\hat x_2)$. Therefore the latter interval must be empty.

The proof is complete for $\hat x\leq x_F^2$. Finally, if $rI^1/(D_{10}-D_{00})<x_F^2<\hat x$, then the ``$\leq$'' in \eqref{hatx} that we derived above for the candidate $x=rI^1/(D_{10}-D_{00})$ must be strict, and thus also ``$<$'' must hold in \eqref{hatx} for $x_F^2$, because otherwise $\hat x\leq x_F^2$ by continuity of $(\beta_1-1)A(x)x^{\beta_1}+(\beta_2-1)B(x)x^{\beta_2}$. Now the verification argument above applies if we consider instead $\hat x:=x_F^2$ with ``$\leq$'' in \eqref{hatx}.
\end{proof}

\begin{proof}[{\bf Proof of Proposition \ref{prop:simeqlPK}}]
The stopping times $\tau_J(\vartheta):=\inf\{t\geq\vartheta\mid x_t\geq x_J\}$, $\vartheta\in\T$, satisfy time consistency $\vartheta'\leq\tau_J(\vartheta)\Rightarrow\tau_J(\vartheta')=\tau_J(\vartheta)$ for any two $\vartheta\leq\vartheta'\in\T$ by construction. $\tau_J(\vartheta)$ is a mutual best reply at $\vartheta$ if the conditions from Proposition \ref{prop:devi} hold. By $x_J\geq x_F^2$, $F^2_{\tau_J(\vartheta)}=M^2_{\tau_J(\vartheta)}$. Under the current specification it suffices to verify conditions \ref{maxM_D} and \ref{maxL_D} for firm $1$.

Condition \ref{maxM_D} holds as waiting until the threshold ${x_J}\leq x_M^1$ is optimal for the constrained problem of stopping $M^1_t$ up to it by Lemma \ref{lem:thresconstr}; cf.\ the unconstrained problem \eqref{sim_i}. Analogously, the threshold $\min\{x_J,x_L^1\}$ solves problem \eqref{monoDi}. Thus condition \ref{maxL_D} holds if $x_L^1\geq x_F^2$ or, using the strong Markov property, if $0\geq D_J(x):=L^1_0-E\bigl[M^1_{\tau_J(0)}\bigr]$ given $x_0=x\in[x_L^1,x_F^2)$.

By Proposition \ref{prop:devi}, if $x_L^1<x_F^2$ solves \eqref{monoDi} and we let $\tau(x)=\inf\{t\geq 0\mid x_t\geq x\}\leq\tau_F^2(0)$ for any $x\in[x_L^1,x_F^2)$, then $D_J(x_L^1)\geq E\bigl[L^1_{\tau(x)}-M^1_{\tau_J(0)}\bigr]=E[D_J(x)]$, where the last identity is due to $x_{\tau(x)}=x$. Therefore it remains to verify $D_J(x_L^1)\leq 0$ for $x_L^1<x_F^2$.

If $x_L^1<x_F^2$, the former is finite and we can write $\lambda:={x_J}/x_L^1\in[1,\infty]$. Then also $x_L^1<x_J$ and thus (cf.\ equations (9), (10) in \cite{PawlinaKort06}, accounting for possibly $x_F^2=\infty$)
\begin{align*}
0\overset{!}{\geq}D_J(x_L^1)&=\frac{x_L^1D_{10}}{r-\mu}-I^1-\frac{x_F^2(D_{10}-D_{11})}{r-\mu}\biggl(\frac{x_L^1}{x_F^2}\biggr)^{\beta_1}\\[4pt]
&-\frac{x_L^1D_{00}}{r-\mu}-\biggl(\frac{{x_J}(D_{11}-D_{00})}{r-\mu}-I^1\biggr)\biggl(\frac{x_L^1}{{x_J}}\biggr)^{\beta_1}\displaybreak[0]\\[4pt]
&=\frac{\beta_1}{\beta_1-1}I^1-I^1-\frac{\beta_1}{\beta_1-1}I^1\frac{D_{10}-D_{11}}{D_{10}-D_{00}}\biggl(\frac{I^1}{I^2}\frac{(D_{11}-D_{01})^+}{D_{10}-D_{00}}\biggr)^{\beta_1-1}\\[4pt]
&-\biggl(\lambda\frac{\beta_1}{\beta_1-1}I^1\frac{D_{11}-D_{00}}{D_{10}-D_{00}}-I^1\biggr)\lambda^{-\beta_1}.
\end{align*}
Rearranging yields condition \eqref{kappa}. The derivative of the square bracket in \eqref{kappa} w.r.t.\ $\lambda$ is strictly negative for $\lambda\in(0,x_M^1/x_L^1)$ given $\beta_1>1$, where it is important to note that $\lambda(D_{11}-D_{00})<D_{10}-D_{00}$, because $D_{10}>D_{00}$ for $x_L^1<x_F^2$ and $(D_{10}-D_{00})/(D_{11}-D_{00})=x_M^1/x_L^1>\lambda$ if $D_{11}>D_{00}$. Using the latter fact also shows that for $\lambda=x_M^1/x_L^1$, the square bracket is either $1-(x_L^1/x_M^1)^{\beta_1}\geq 0$ or $1$, if $x_M^1$ is finite or not, respectively.

Finally, necessity of $D_J(x_L^1)\leq 0$ for $x_L^1<x_F^2\leq x_J$ is obvious.
\end{proof}

\begin{proof}[{\bf Proof of Proposition \ref{prop:seqeqlPK}}]
By the hypothesis $x_L^1<x_F^2$ and Lemmas \ref{lem:stoptau^1_F} and \ref{lem:thresconstr}, problem \eqref{maxL_S} is solved by $\tau_S(\vartheta):=\tau_L^1(\vartheta)=\inf\{t\geq\vartheta\mid x_t\geq x_L^1\}\in\T$ for any $\vartheta\in\T$. These stopping times for firm $1$ satisfy time consistency $\vartheta'\leq\tau_S(\vartheta)\Rightarrow\tau_S(\vartheta')=\tau_S(\vartheta)$ for any two $\vartheta\leq\vartheta'\in\T$ by construction, as do firm $2$'s stopping times $\tau_F^2(\vartheta)=\inf\{t\geq\vartheta\mid x_t\geq x_F^2\}$.

To verify the equilibrium at $\vartheta\in\T$ by Corollary \ref{cor:seqeql}, note that now $\pi^{L1}_\cdot-\pi^{01}_\cdot\geq\pi^{L2}_\cdot-\pi^{02}_\cdot$, whence problem \eqref{monoDi} is solved by $\tau_D^2(\vartheta')=\tau_S(\vartheta)\vee\vartheta'$. Thus we have an equilibrium if $x_L^1\geq x_F^1$ ($\geq\bar x$) or, using the strong Markov property, if $0\geq D_S(x):=L^2_0-E\bigl[F^2_{\tau_S(0)}\bigr]$ given $x_0=x\in[x_L^1,x_F^1)$.

By Proposition \ref{prop:devi}, if $x_L^1<x_F^1$ and we let $\tau(x)=\inf\{t\geq 0\mid x_t\geq x\}\leq\tau_F^1(0)$ for any $x\in[x_L^1,x_F^1)$, then $D_S(x_L^1)\geq E\bigl[L^2_{\tau(x)}-F^2_{\tau_S(0)}\bigr]=E[D_S(x)]$, where the last identity is due to $x_{\tau(x)}=x$. Therefore it remains to verify $D_S(x_L^1)\leq 0$ for $x_L^1<x_F^1$, i.e., $x_L^1\not\in(\ubar x,\bar x)$. The latter condition is (cf.\ equations (8), (9) in \cite{PawlinaKort06}, accounting for possibly $x_F^1=x_F^2=\infty$)
\begin{align*}
0\overset{!}{\geq}D_S(x_L^1)&=\frac{x_L^1D_{10}}{r-\mu}-I^2-\frac{x_F^1(D_{10}-D_{11})}{r-\mu}\biggl(\frac{x_L^1}{x_F^1}\biggr)^{\beta_1}\\[4pt]
&-\frac{x_L^1D_{01}}{r-\mu}-\biggl(\frac{x_F^2(D_{11}-D_{01})}{r-\mu}-I^2\biggr)\biggl(\frac{x_L^1}{x_F^2}\biggr)^{\beta_1}\\[4pt]
&=\frac{\beta_1}{\beta_1-1}I^1\frac{D_{10}-D_{01}}{D_{10}-D_{00}}-I^2-\frac{\beta_1}{\beta_1-1}I^1\frac{D_{10}-D_{11}}{D_{10}-D_{00}}\biggl(\frac{(D_{11}-D_{01})^+}{D_{10}-D_{00}}\biggr)^{\beta_1-1}\\[4pt]
&-\frac{1}{\beta_1-1}I^2\biggl(\frac{I^1}{I^2}\frac{(D_{11}-D_{01})^+}{D_{10}-D_{00}}\biggr)^{\beta_1}.
\end{align*}
Rearranging yields condition \eqref{kappaseql}. The derivative of its LHS w.r.t.\ $I^2/I^1$ is strictly positive for $x_L^1<x_F^1$ given $\beta_1>1$, because then $(D_{11}-D_{00})^+/(D_{10}-D_{00})<1$. By the same fact the RHS of \eqref{kappaseql} is strictly positive.

To show necessity of $x_L^1\not\in(\ubar x,\bar x)$, suppose the contrary, whence $x_L^1<x_F^1$ and $D_S(x_L^1)>0$ by definition. For any $x\leq x_L^1$,
\begin{align*}
D_S(x)&=E\Bigr[D_S(x_L^1)\Bigr]+L^2_0-E\Bigr[L^2_{\tau_S(0)}\Bigr]=D_S(x_L^1)+E\biggl[\int_0^{\tau_S(0)}(\pi^{L2}_s-\pi^{02}_s)\,ds\biggr] \\
&=D_S(x_L^1)+\frac{x(D_{10}-D_{00})}{r-\mu}-I^2-\frac{x_L^1(D_{10}-D_{00})}{r-\mu}\biggl(\frac{x}{x_L^1}\biggr)^{\beta_1},
\end{align*}
which converges continuously to $D_S(x_L^1)>0$ as $x\to x_L^1$. Thus $D_S(x)>0$ for some $x<x_L^1$.
\end{proof}

\singlespacing
\bibliography{jhs}
\end{document}